\documentclass[12pt]{article}
\usepackage{amsmath}
\usepackage{amssymb}
\usepackage{amsfonts}
\usepackage{graphicx}
\usepackage{enumitem}
\usepackage{natbib}
\setcitestyle{round}
\renewcommand{\mathbf}[1]{\boldsymbol{#1}}
\usepackage{authblk}

\newtheorem{theorem}{Theorem}

\newtheorem{lemma}{Lemma}
\newtheorem{remark}{Remark}
\newenvironment{proof}[1][Proof]{\noindent\textbf{#1.} }{\ \rule{0.5em}{0.5em}}

\begin{document}
\title{Robust estimators of accelerated failure time regression with generalized log-gamma errors}
\author[1]{Claudio Agostinelli}
\author[2]{Isabella Locatelli}
\author[2]{Alfio Marazzi\thanks{Corresponding author alfio.marazzi@chuv.ch}}
\author[3]{V{\'\i}ctor J. Yohai}

\affil[1]{Department of Mathematics, University of Trento, Trento, Italy}
\affil[2]{Institute of social and preventive medicine, Lausanne University Hospital, Switzerland}
\affil[3]{Departamento de Matematicas, Facultad de Ciencias Exactas y Naturales, University of Buenos Aires and CONICET}

\date{\today}

\maketitle

\begin{abstract}
The generalized log-gamma (GLG) model is a very flexible family of distributions to analyze datasets in many different areas of science and technology. In this paper, we propose estimators which are simultaneously highly robust and highly efficient for the parameters of a GLG distribution in the presence of censoring. We also introduced estimators with the same properties for accelerated failure time models with censored observations and error distribution belonging to the GLG family. We prove that the proposed estimators are asymptotically fully efficient and examine the maximum mean square error using Monte Carlo simulations. The simulations confirm that the proposed estimators are highly robust and highly efficient for finite sample size. Finally, we illustrate the good behavior of the proposed estimators with two real datasets.

\end{abstract}

\section{Introduction} 
\label{SecIntroduction} 
Generalized log-gamma (GLG) regression with censored observations is a large subclass of Accelerated Failure Time (AFT) models introduced by \citet{lawless1980}. Many models broadly used in the lifetime data analysis -- including lognormal, log-gamma, and log-Weibull regression -- are specific cases of GLG regression. GLG regression has been widely applied in various areas of survival analysis \citep[e.g.][]{kim1993,sun1999,abadi2012}.

Usually, the parameters are estimated by means of the maximum likelihood (ML) principle, which provides fully efficient estimators when the observations follow the model. Unfortunately, the ML estimator is extremely sensitive to the presence of outliers in the sample.

There are several proposals of diagnostic tools to detect outliers and assess their influence on the GLG regression parameter estimates \citep[e.g.][]{ortega2003,ortega2008,silva2010}. Moreover, two families of robust estimators of the GLG model without censored observations and without covariate information have been introduced by \citet{agostinelli2014a} for models with three parameters: location, scale, and shape. These families of estimators are: the (weighted) quantile $\tau$ (Q$\tau$) estimators and the one-step weighted likelihood (1SWL) estimators. A Q$\tau$ estimator minimizes a $\tau$ scale \citep{yohai1988} of the differences between empirical and theoretical quantiles. It is $n^{1/2}$ consistent but not asymptotically normal. However, it is a convenient starting point to define the 1SWL-estimator.

In this paper, we extend the Q$\tau$ estimator proposed in \citet{agostinelli2014a} to GLG regression with right censoring by introducing the trimmed Q$\tau$-estimator (TQ$\tau$-estimator); we also extend the truncated maximum likelihood (TML) estimator introduced in \citet{Marazzi2004} to GLG regression and. To improve the robustness of this estimator without modifying its asymptotic efficiency we also define a one-step version of the TML estimator (1TML-estimator). For the sake of completeness, we also define an extension of the 1SWL estimator which is fully described in the Appendix. However a Monte Carlo study show that this estimator is much less robust than the TQ$\tau$- and 1TML- estimators.

The procedures introduced here for the GLG family can be applied to other location-scale-shape models, such as the three-parameter log-Weibull family.

Section \ref{SecEstimatorsCWC} defines the Q$\tau$- and TQ$\tau$-estimators for censored observations in the absence of covariates. Section \ref{SecTML} describes the TML estimators. Section \ref{SecRegression} extends the estimators to the regression case. Section \ref{SecMonteCarlo} shows the results of a Monte Carlo study comparing the performance of the proposed methods for finite sample sizes. Section \ref{SecExamples} discusses two examples with real data. Section \ref{SecConclusions} provides concluding remarks. 

Section \ref{appendix} is an Appendix that include the proofs and an extension for censored data of the one step weighted likelihood estimator used in \citet{agostinelli2014a}. For completeness sake, this estimator is included in our Monte Carlo study. 

%%%%%Proofs and complementary technical details can be found in the Appendix.

\section{The Q$\tau$- and TQ$\tau$-estimators for censored observations without covariates} 
\label{SecEstimatorsCWC}

\subsection{The generalized log-gamma distribution}
The GLG family of distributions depends on three parameters $\mu$, $\sigma$, and $\lambda$. We use the parametrization of Prentice (1974) and denote the family by $GLG(\mu,\sigma,\lambda)$, $\mu\in\mathbb{R}$, $\sigma>0$, $\lambda\in\mathbb{R}$. A random variable $y$ has a $GLG(\mu,\sigma,\lambda)$ distribution if
\begin{equation}
y=\mu+\sigma u \label{locscal1}
\end{equation}
and $u$ has density
\begin{equation}
f_{\lambda}(u)=\left\{
\begin{array}[c]{lll}
\frac{\left\vert \lambda\right\vert }{\Gamma\left(  \lambda^{-2}\right) }(\lambda^{-2})^{\lambda^{-2}}\exp\left(  (\lambda^{-2})\left(  \lambda u-e^{\lambda u}\right)  \right)  & \text{if} & \lambda\neq0,\\ 
\frac{1}{\sqrt{2\pi}}\exp(-\frac{u^{2}}{2}) & \text{if} & \lambda=0,
\end{array}
\right. \label{flambda} 
\end{equation}
where $\Gamma$ denotes the Gamma function. This family includes many common models, such as the log-Weibull model ($\lambda=1$), the log-exponential model ($\lambda=1$ and $\sigma=1$), the log-gamma model ($\sigma=\lambda$), and the normal model ($\lambda=0$). The density of $y$ is
\begin{equation}
f_{\mathbf{\theta}}\left(  y\right)  =\frac{1}{\sigma}f_{\lambda}\left( \frac{y-\mu}{\sigma}\right)  ,\label{denggd}
\end{equation}
where $\mathbf{\theta}=(\mu,\sigma,\lambda)$.

Suppose that $y_{1},\ldots,y_{n}$ are $n$ i.i.d. random times with a GLG cdf $F_{\mathbf{\theta}_{0}}(y)$. We want to estimate $\mathbf{\theta}_{0}=(\mu_{0},\sigma_{0},\lambda_{0})$. We consider single censoring on the right, i.e., the true value of $y_{i}$ is not observed. Instead, the censored variable $y_{i}^{\ast}=\min(y_{i},c_{i})$ is observed, where $c_{1},\ldots,c_{n}$ are i.i.d. censoring times, which are independent of the $y_{i}$'s. We define the censoring indicator $\delta_{i}=1$ if $y_{i}^{\ast}=y_{i}$ and $\delta_{i}=0$ if $y_{i}^{\ast}=c_{i}$. Let $\mathbf{z}_{i}=(y_{i}^{\ast},\delta_{i})$ and $G_{n}$ be the empirical distribution function based on $(\mathbf{z}_{1},\ldots,\mathbf{z}_{n})$.

\subsection{Score functions and ML estimator}
\label{SecML} 
The ML estimator of the parameters of a GLG model under censoring can be easily defined as follows. Let $S_{\mathbf{\theta}}(y)=1-F_{\mathbf{\theta}}(y)=1-F_{\lambda}((y-\mu)/\sigma)$ denote the survival function. Then, the negative log-likelihood function is
\begin{equation}
\label{equNegLogLik}-\sum_{i=1}^{n} \left[  \delta_{i}\log f_{\mathbf{\theta}}(y_{i}^{\ast})+(1-\delta_{i})\log S_{\mathbf{\theta}}(y_{i}^{\ast})\right]  .
\end{equation}
In the absence of censoring, the score functions $\mathbf{d}=(d_{1}, d_{2}, d_{3})^{\top}$ are:
\begin{align}
d_{1}(y,\mathbf{\theta)}  &  =-\frac{\partial}{\partial\mu}\log f_{\mathbf{\theta}}(y)=\frac{1}{\sigma}\xi_{\lambda}(u),\label{zscore1}\\
d_{2}(y,\mathbf{\theta)}  &  =-\frac{\partial}{\partial\sigma}\log f_{\mathbf{\theta}}(y)=\frac{1}{\sigma}\left(  \xi_{\lambda}(u)u+1\right), \label{zscore2} \\
d_{3}(y,\mathbf{\theta)}  &  =-\frac{\partial}{\partial\lambda}\log f_{\mathbf{\theta}}(y)=\psi_{\lambda}(u),\label{zscore3}
\end{align}
where $u=(y-\mu)/\sigma$, $\dot{F}_{\mathbf{\theta}}(y)=\dot{F}_{\lambda}(u)=dF_{\lambda}(u)/d\lambda$,
\begin{align*}
\xi_{\lambda}(u)  &  =\frac{f_{\lambda}^{\prime}\left(  u\right)  }{f_{\lambda}\left(  u\right)  }=\frac{(1-e^{\lambda u})}{\lambda}, \\
\psi_{\lambda}(u)  &  =-\frac{\partial}{\partial\lambda}\log f_{\lambda
}(u)=\frac{1}{\lambda^{3}}(2\zeta(\lambda)-\lambda^{2}+\lambda u-\exp(\lambda u)(2-\lambda u)),
\end{align*}
$\zeta(\lambda)=-2\log(\lambda)-\dot{\Gamma}(\lambda^{-2})+1$, and $\dot{\Gamma}$ denotes the digamma function. Let
\begin{align}
s_{1}(y,\mathbf{\theta)}  &  =-\frac{\partial}{\partial\mu}\log S_{\mathbf{\theta}}(y)=-\frac{1}{\sigma}\frac{f_{\lambda}(u)}{S_{\lambda}(u)},\label{sscore1} \\
s_{2}(y,\mathbf{\theta)}  &  =-\frac{\partial}{\partial\sigma}\log S_{\mathbf{\theta}}(y)=-\frac{1}{\sigma}\frac{f_{\lambda}(u)u}{S_{\lambda}(u)}=s_{1}(y,\mathbf{\theta)}u,\label{sscore2}\\
s_{3}(y,\mathbf{\theta)}  &  =-\frac{\partial}{\partial\lambda}\log S_{\mathbf{\theta}}(y)=\frac{\dot{F}_{\lambda}(u)}{S_{\lambda}(u)} . \label{sscore3}
\end{align}
Then, the score functions for the case with censored observations are
\begin{equation}
v_{k}(y,\delta,\mathbf{\theta}) = \delta d_{k}(y,\mathbf{\theta})+(1-\delta)s_{k}(y,\mathbf{\theta}), \qquad k=1,2,3. \label{vscore}
\end{equation}
The \textit{ML estimator} of $\mathbf{\theta}$ is given by the following system of equations
\begin{equation}
E_{G_{n}}(\mathbf{v}(y,\delta,\mathbf{\theta})) = \frac{1}{n}\sum_{i=1}^{n} \mathbf{v}(y_{i}^{\ast},\delta_{i},\mathbf{\theta})=\mathbf{0} ,\label{equML1}
\end{equation}
where $\mathbf{v}=(v_{1},v_{2},v_{3})^{\top}$ is the score function vector. It is easy to show that
\begin{equation*}
\nabla_{\mathbf{\theta}} S_{\mathbf{\theta}}(y^{\ast}) = E_{\mathbf{\theta}%
}(\mathbf{d}(y,\mathbf{\theta})| y > y^{\ast}),
\end{equation*}
where $\nabla_{\mathbf{\theta}}$ indicates differentiation w.r.t.
$\mathbf{\theta}$. Hence, an alternative expression for the likelihood
equations is
\begin{equation*}
\frac{1}{n}\sum_{i=1}^{n} \delta_{i} \mathbf{d}(y_{i}^{\ast},\mathbf{\theta}) + (1 - \delta_{i}) E_{\mathbf{\theta}}(\mathbf{d}(y,\mathbf{\theta})| y > y_{i}^{\ast}) =\mathbf{0}.
\end{equation*}
Following \citet{locatelli2010}, we define the \textit{semiempirical cdf of $y$ for a given $\mathbf{\theta}$} as
\begin{equation*}
H_{n,\mathbf{\theta}}(z)=\frac{1}{n}\sum_{i=1}^{n}E_{\mathbf{\theta}}\left[ \left.  I\left(  y\leq z\right)  \right\vert y_{i}^{\ast},\delta_{i} \right]
\end{equation*}
or equivalently,
\begin{equation}
H_{n,\mathbf{\theta}} \left(  z\right)  =\frac{1}{n} \sum_{i=1}^{n} \delta_{i} I\left(  y_{i}\leq z\right)  +\frac{1}{n} \sum_{i=1}^{n} \left(  1-\delta_{i}\right)  \frac{\left[  F_{\mathbf{\theta}}\left(  z\right) -F_{\mathbf{\theta}}(y_{i}) \right]  ^{+}}{1-F_{\mathbf{\theta}}(y_{i})}. \label{Hemp}
\end{equation}
Thus, when there is no censoring, $H_{n,\mathbf{\theta}}(z)$ coincides with the usual empirical cdf. If $\tilde{\mathbf{\theta}}$ is a consistent estimator of $\mathbf{\theta}_{0}$, then $H_{n,\tilde{\mathbf{\theta}}}(z)$ is a consistent estimator of $F_{\mathbf{\theta}_{0}}(z)$ and, for any measurable function $h(y)$, we have $\lim_{n\rightarrow\infty} E_{n,\tilde{\mathbf{\theta}}} \left[  h\left(  y\right)  \right]  =E_{\mathbf{\theta}_{0}}\left[ h\left(  y\right)  \right]  $ a.s., where $E_{n,\mathbf{\theta}}$ denotes expectation under $H_{n,\mathbf{\theta}}$. Finally, another expression of the likelihood equations is
\begin{equation}
\label{equML2}E_{n,\mathbf{\theta}}(\mathbf{d}(y,\mathbf{\theta})) = \mathbf{0} .
\end{equation}
Let $\mathbf{M}(\mathbf{\theta}) = E\left(  \mathbf{v}(y,\delta,\mathbf{\theta}) \mathbf{v}(y,\delta,\mathbf{\theta})^{\top}\right)$ and $\mathbf{G} (\mathbf{\theta}) = E\left(  \nabla_{\mathbf{\theta}} \mathbf{v}(y,\delta,\mathbf{\theta}) \right)  $ then the asymptotic covariance matrix of the ML estimator is
\begin{equation*}
\mathbf{\Sigma}(\mathbf{\theta}) = \mathbf{G}(\mathbf{\theta})^{-1} \mathbf{M}(\mathbf{\theta}) \mathbf{G}(\mathbf{\theta})^{-\top} .
\end{equation*}

\subsection{The trimmed Q$\tau$ estimator}
\label{SecTQtau} 
\citet{agostinelli2014a} define the quantile $\tau$ (Q$\tau$) estimator and the weighted Q$\tau$ (WQ$\tau$) estimator for non-censored i.i.d. observations as follows. For $0<u<1$ let $Q(u,\mathbf{\theta})$ denote the $u$-quantile of $F_{\mathbf{\theta}}(y)$. Then, $Q(u,\mathbf{\theta})=\sigma Q^{\ast}(u,\lambda)+\mu$, where $Q^{\ast}(u,\lambda)=Q(u,(0,1,\lambda))$. Given a sample $y_{1},\ldots,y_{n}$, let $F_{n}$ denote the empirical cdf of $y$. Then, $y_{(1)},\ldots,y_{(n)}$, the ordered observations, are the quantiles $u_{n,j}=(j-0.5)/n$ of $F_{n}$ and should be close to $\sigma_{0}Q^{\ast}(u_{n,j}, \lambda_{0})+\mu_{0}$ for $j=1,\ldots,n$. Consider the differences between the empirical and the theoretical quantiles
\begin{equation*}
r_{n,j}(\mathbf{\theta})=y_{(j)}-\mu-\sigma Q^{\ast}(u_{n,j},\lambda), \qquad j=1,\ldots,n.
\end{equation*}
The \textit{Q$\tau$ estimator} is defined by
\begin{equation*}
\mathbf{\tilde{\theta}}_{n}=\arg\min_{\mathbf{\theta}} \tau(r_{n,1} (\mathbf{\theta}),\ldots,r_{n,n}(\mathbf{\theta})),
\end{equation*}
where $\tau$ denotes the $\tau$ scale.

The $\tau$ scale was introduced by \citet{yohai1988} to define estimators which combine high finite sample breakdown point with high efficiency in the linear model with normal errors. Given a sample $\mathbf{u}=(u_{1},\ldots,u_{n})$, a function $s(\mathbf{u})$ is called a \textit{scale} if: (i) $s(\mathbf{u}) \geq0$; (ii) for any scalar $\gamma$, $s(\gamma\mathbf{u})=|\gamma|s(\mathbf{u})$; (iii) $s(u_{1},\ldots,u_{n})=s(|u_{1}|,\ldots,|u_{n}|)$; (iv) if $|u_{i}|\leq|v_{i}|$, $1\leq i\leq n$, then $s(u_{1},\ldots,u_{n}) \leq s(v_{1},\ldots,v_{n})$. It follows that (v) $s(0,\ldots,0)=0$ and that, (vi) given $\varepsilon>0$, there exists $\delta$ such that $|u_{i}|\leq\delta$ for $1\leq i\leq n$ imply $s(u_{1},\ldots,u_{n})<\varepsilon$. Properties (i)-(vi) clearly show that $s(\mathbf{u})$ can be used as a measure of the absolute largeness of the elements of $\mathbf{u}$. The most common scale is the one based on the quadratic function and is given by $s_{1}(\mathbf{u})=\left(  \sum_{i=1}^{n}u_{i}^{2}/n\right)  ^{1/2}$. This scale is clearly non robust. \citet{huber1981} defines a general class of robust scales, called M scales, as follows. Let $\rho$ be a function satisfying the following properties:

\begin{enumerate}[label=\textbf{A\arabic*}]
\item \label{assA1}: (i) $\rho(0)=0$; (ii) $\rho$ is even; (iii) if $|x_{1}|<|x_{2}|$, then $\rho(x_{1})\leq\rho(x_{2})$; (iv) $\rho$ is bounded; (v) $\rho$ is continuous.
\end{enumerate}
Then, an \textit{M scale} $s_{2}(\mathbf{u})$ based on $\rho$ is defined by the value $s$ satisfying
\begin{equation}
\label{Mscale}\frac{1}{n}\sum_{i=1}^{n} \rho\left(  \frac{u_{i}}{s} \right) = b ,
\end{equation}
where $b$ is a given scalar and $0<b<a=\sup\rho$. \citet{yohai1988} introduce the family of $\tau$ scales. A $\tau$ \textit{scale} is based on two functions $\rho_{1}$ and $\rho_{2}$ satisfying conditions \ref{assA1} and such that $\rho_{2}\leq\rho_{1}$. One considers an M scale $s_{2}(\mathbf{u})$ defined by (\ref{Mscale}) with $\rho_{1}$ in place of $\rho$; then, the \textit{$\tau$ scale} is given by
\begin{equation}
\label{tauscale}\tau^{2}(\mathbf{u})=s_{2}^{2}(\mathbf{u})\frac{1}{n} \sum_{i=1}^{n} \rho_{2}\left(  \frac{u_{i}}{s_{2}(\mathbf{u)}}\right) .
\end{equation}
Usually, $\rho$, $\rho_{1}$ and $\rho_{2}$ are selected in the Tukey's
bisquare family given by
\begin{equation}
\label{equTukeyRho}\rho_{c}^{T}(u) = 1-\max\left(  \left(  1-\left( \frac{u}{c}\right)  ^{2}\right)  ^{3}, 1 \right)
\end{equation}
for convenient values of $c$ and $b$ (see \ref{SecMonteCarlo}).

The Q$\tau$ estimator in the case of randomly censored observations is obtained by replacing the quantiles of the empirical distribution by the quantiles of the Kaplan-Meier (KM) distribution corresponding to the non censored observations. More precisely, let $\tilde{F}_{n}$ denote the KM estimator \citep{kaplan1958} of $F_{\mathbf{\theta}_{0}}$ and $t_{(1)},\ldots,t_{(m)}$ the ordered non censored observations. Then, $t_{(1)},\ldots,t_{(m)}$ are the quantiles $\tilde{u}_{n,i}=\tilde{F}_{n}(t_{(i)})-0.5/n$. The residuals
\begin{equation*}
\tilde{r}_{n,i}(\mathbf{\theta})=t_{(i)}-\mu-\sigma Q^{\ast}(\tilde{u}_{n,i},\lambda), \qquad i=1,\ldots,m.
\end{equation*}
are then used to define the \textit{Q$\tau$ estimator for censored observations} by
\begin{equation}
\mathbf{\tilde{\theta}}_{n}=\arg\min_{\mathbf{\theta}} \tau(\tilde{r}_{n,1}(\mathbf{\theta}),\ldots,\tilde{r}_{n,m}(\mathbf{\theta})). \label{QTAU}
\end{equation}
It is known that the KM estimator distributes the mass of the censored observations among all the observations that are on their right. Some of these observations may be outliers and, therefore, the mass assigned to the outliers by KM may be inflated. To reduce the influence of outliers we defined a trimmed version of the Q$\tau$ estimator as follows. Let $0<\alpha<1$ be the fraction of trimming and let $k_{\alpha}=\max\{h:\tilde{F}_{n}(t_{(h)})\leq1-\alpha\}$ then the $\alpha$-\textit{trimmed Q}$\tau$ ($\alpha$-TQ$\tau$) \textit{estimator} is defined by
\begin{equation}
\mathbf{\tilde{\theta}}_{n}=\arg\min_{\mathbf{\theta}}\tau(\tilde{r}_{n,1}(\mathbf{\theta}),\ldots,\tilde{r}_{n,k_{\alpha}}(\mathbf{\theta})). \label{TQTAU}
\end{equation}
Since in all the simulations and examples we use $\alpha=0.1,$ to simplify the notation in the remaining of the paper, we will write TQ$\tau$ instead of $\alpha$-TQ$\tau$.

Note that the residuals $\tilde{r}_{n,i}(\mathbf{\theta})$ are heteroskedastic and, according to \citet{serfling1980}, their variance can be approximated by
\begin{equation}
\hat{\sigma}_{i}^{2}=v_{\tilde{u}_{n,i}}^{\ast2}/f_{\tilde{\lambda}}^{2}(Q^{\ast}(\tilde{u}_{n,i},\tilde{\lambda})),\label{GH}
\end{equation}
where $v_{\tilde{u}_{n,i}}^{\ast2}$ is Greenwood's variance estimator of $\tilde{F}_{n}(t_{(i)})$ \citep{greenwood1926}. Then, as in \citet{agostinelli2014a}, we might consider the \textit{weighted} TQ$\tau$ (WTQ$\tau$) \textit{estimator}
\begin{equation}
\mathbf{\tilde{\theta}}_{n}^{w}=\arg\min_{\mathbf{\theta}}\tau\left( \frac{\tilde{r}_{n,1}(\mathbf{\theta})}{\hat{\sigma}_{1}},\ldots,\frac{\tilde{r}_{n,k_{\alpha}}(\mathbf{\theta})}{\hat{\sigma}_{k_{\alpha}}}\right) . \label{WQTAU}
\end{equation}
However, our Monte Carlo experiments have shown that weighting does not provide an important improvement and, for this reason, we are not going to consider the WTQ$\tau$ estimator further.

In Theorem \ref{TeoSMConsistencyTQtau} of the Appendix we prove that, under general conditions, the TQ$\tau$ estimator $\mathbf{\tilde{\theta}}_{n}$ is $n^{1/2}$-consistent, that is,
\begin{equation}
n^{1/2}(\mathbf{\tilde{\theta}}_{n}-\mathbf{\theta}_{0})=O_{p}(1) . \label{sqrnc}
\end{equation}
As in the non censored case, one drawback of the quantile $\tau$ estimators is that they are not asymptotically normally distributed, making inference difficult. In order to overcome this problem we introduce, in the next Section, the one step truncated ML estimator starting at TQ$\tau$. This estimator has similar robustness properties as the TQ$\tau$ estimator; in addition, it has asymptotic normal distribution with the same asymptotic variance as the ML estimator under the model.

\section{The truncated maximum likelihood estimators}
\label{SecTML} 
In order to obtain a robust and highly efficient procedure for estimating the unknown parameter vector we use the truncated maximum likelihood (TML) estimator and the one-step TML (1TML) estimator. Both are based on a weighted form of the likelihood equations. Suppose that $y\sim GLG(\mathbf{\theta}_{0})$ and consider a sample $(\mathbf{z}_{1},\ldots,\mathbf{z}_{n})$, $z_{i}=(y_{i}^{\ast},\delta_{i})$ ($i=1,\ldots,n$). Assume that $\tilde{\mathbf{\theta}}=(\tilde{\mu},\tilde{\sigma},\tilde{\lambda})$ is an initial consistent estimator of $\mathbf{\theta}_{0}$. A natural robustification of the likelihood equations (\ref{equML2}) can be obtained by weighting the equations. More precisely, given a weight function $w(y,\mathbf{\theta})$, we consider the equations
\begin{equation}
E_{n,{\mathbf{\theta}}}[w(y,\tilde{\mathbf{\theta}})\mathbf{d}(y,\mathbf{\theta})]=E_{n,\tilde{\mathbf{\theta}}}[w(y,\tilde{\mathbf{\theta}})\mathbf{d}(y,\tilde{\mathbf{\theta}})] . \label{weq}
\end{equation}
The right hand side mitigates the bias of the estimator and allows to proof asymptotic normality. For increasing sample size its tends to zero. In the next subsections we show how one can define the weight function.

\subsection{The outlier rejection rule}
We proceed as in \citet{Marazzi2004}. Let $\tilde{r}_{i}^{\ast}=(y_{i}^{\ast}-\tilde{\mu})/\tilde{\sigma}$ denote the standardized residuals with respect to the initial model and let $l_{\lambda}(z)=-\log f_{\lambda}(z)$ be the negative log-likelihood function. We consider the negative log-likelihoods of the residuals $l_{i}^{\ast}=l_{\tilde{\lambda}}(\tilde{r}_{i}^{\ast})$ ($i=1,\ldots,n$). A large $l_{i}^{\ast}$ corresponds to an observation with a small likelihood under the model and suggests that $y_{i}$ is an outlier. Let $M_{\lambda}$ be the cdf of $l_{\lambda}(z)$ and
\begin{equation*}
M_{n,\tilde{\lambda}}(z)=\frac{1}{n}\sum_{i=1}^{n}\left[  \delta_{i}I\left( l_{i}^{\ast}\leq z\right)  +(1-\delta_{i})P_{\tilde{\lambda}}(l_{\tilde{\lambda}}\leq z|y>y_{i}^{\ast})\right] ,
\end{equation*}
be the semiempirical cdf of $l_{1}^{\ast},\ldots,l_{n}^{\ast}$ for $\lambda=\tilde\lambda$. One can show that $M_{n,\tilde{\lambda}}$ is a consistent estimator of $M_{\lambda_{0}}$. For simplicity, we write $M_{n}$ in place of $M_{n,\tilde{\lambda}}$. Let $M_{n}^{(\varphi)}$ denote $M_{n}$ truncated at $\varphi$, i.e.,
\begin{equation}
M_{n}^{(\varphi)}(t)=\left\{
\begin{array}[c]{cc}
M_{n}(t)/M_{n}(\varphi) & \text{if} \ t \leq\varphi,\\
1 & \text{otherwise}.
\end{array}
\right. \label{trunc}
\end{equation}
We want to compare the right tail of the truncated empirical distribution $M_{n}^{(\varphi)}$ with the right tail of $M_{n,\tilde{\lambda}}(t)$, which is the theoretical distribution when $\lambda=\tilde{\lambda}$. To specify what we understand by the tail, we take a number $\varepsilon$ close to $0$, for example $\varepsilon=0.01$, as the probability of falling in the tail. Then, we define the cutoff point $\varphi^{\ast}$ on the likelihood scale as the largest $\varphi$ such that $M_{n}^{(\varphi)}(t)\geqq M_{\tilde{\lambda}}(t)$ for all $t\geq M_{\tilde{\lambda}}^{-1}(1-\varepsilon)$, i.e.,
\begin{equation*}
\varphi^{\ast} = \sup\{ \varphi| M_{n}^{(\varphi)}(t) \geqq M_{\tilde{\lambda}}(t)\text{ for all }t\geq M_{\tilde{\lambda}}^{-1}(1-\varepsilon) \}.
\end{equation*}
Note that $\varphi^{\ast}$ is the minimum value $\varphi$ such that the tails of $M_{n}^{(\varphi)}(t)$ and $M_{\tilde{\lambda}}(t)$ are comparable. As in \citet{gervini2002}, one can prove that, if the sample does not contain outliers, $\varphi^{\ast} \rightarrow\infty$ a.s. as $n\rightarrow\infty$. Finally, since $l_{\tilde{\lambda}}(z)$ is unimodal, there exist two solutions $\tilde{c}_{L}$ and $\tilde{c}_{U}$ of the equation $l_{\tilde{\lambda}}(z) = \varphi^{\ast}$. It is immediate that $l_{\tilde{\lambda}}(z) \le\varphi^{\ast}$ is equivalent to $\tilde{c}_{L} \le z \le\tilde{c}_{U}$. The cutoff points on the data scale are $\tilde{t}_{L}=\tilde{\mu}+\tilde{\sigma} \tilde{c}_{L}$ and $\tilde{t}_{U}=\tilde{\mu}+\tilde{\sigma}\tilde{c}_{U}$.

\subsection{Weight functions}
Let $\omega(z)$ be a function, such that

\begin{enumerate}[label=\textbf{A\arabic*}] \setcounter{enumi}{1}
\item \label{assA2} (i) $\omega(z)$ is nonincreasing; (ii) $\lim_{z\rightarrow-\infty} \omega(z)=1$; (iii) $\omega(z)=0$ for $z>0$.
\end{enumerate}
For example, let $c>0$ and consider the function
\begin{equation}
\label{wgtfun}\omega(z)=\rho(z,c) \cdot I(z \leq0),
\end{equation}
where $\rho(z,c) = \rho^{T}(z/c)$ is in the biweight family as defined in (\ref{equTukeyRho}) and let $\vartheta^{\ast} = 1/\varphi^{\ast}$. Then, define the weight function
\begin{equation*}
w(z, \lambda, \vartheta^{\ast}) = \omega\left(  l_{\lambda}(z) - 1/\vartheta^{\ast} \right)
\end{equation*}
or, for the observation $y$,
\begin{equation*}
w(y, \mathbf{\theta}, \vartheta^{\ast}) = w\left(  \frac{y-\mu}{\sigma}, \lambda, \vartheta^{\ast}\right)  .
\end{equation*}
Note that for $c \rightarrow0$, we have $\omega(z)=I(z\leq0)$. In this case, $w_{i}=1$ if $y_{i} \epsilon[\tilde{t}_{L},\tilde{t}_{U}]$ and $w_{i}=0$ otherwise; this rule is usually called \textit{hard rejection}.

\subsection{The TML estimators}
We suppose that an initial estimator $\tilde{\mathbf{\theta}}=(\tilde{\mu},\tilde{\sigma},\tilde{\lambda})$ and a cutoff point $\vartheta^{\ast}$ are given. The \textit{TML estimator} $\hat{\mathbf{\theta}}=(\hat{\mu},\hat{\sigma},\hat{\lambda})$ is the solution of the equations
\begin{equation}
\mathbf{u}(\mathbf{\theta},\tilde{\mathbf{\theta}},\vartheta^{\ast})=\tilde{\mathbf{h}},\label{equTML}
\end{equation}
where $\mathbf{u}=(u_{1},u_{2},u_{3})^{\top}$ is given by
\begin{align*}
\mathbf{u}(\mathbf{\theta},\tilde{\mathbf{\theta}},\vartheta^{\ast}) & = E_{n,\tilde{\mathbf{\theta}}}\left[  w(y,\tilde{\mathbf{\theta}},\vartheta^{\ast})\mathbf{d}(y,\mathbf{\theta})\right] \\
&  =\frac{1}{n}\sum_{i=1}^{n}E_{\mathbf{\theta}}\left[  w\left(y,\tilde{\mathbf{\theta}},\vartheta^{\ast}\right)  \mathbf{d}\left(y,\mathbf{\theta}\right)  |y_{i}^{\ast},\delta_{i}\right]
\end{align*}
and
\begin{equation}
\tilde{\mathbf{h}}=(\tilde{h}_{1},\tilde{h}_{2},\tilde{h}_{3})^{\top}=E_{\tilde{\lambda}}\left[  w(u,\tilde{\lambda},\vartheta^{\ast})\mathbf{d}(u)\right] \label{htilde}
\end{equation}
with $u\sim f_{(0,1,\tilde{\lambda})}$. These equations are of the form (\ref{weq}). Note that, when $\vartheta^{\ast}\rightarrow0$, the right hand side vector tends to zero and we obtain the ML equations.

The \textit{1TML estimator} is obtained by applying one iteration of the Newton-Raphson procedure to equations (\ref{equTML}) and it turns out to be
\begin{equation}
\label{equ1TML}\hat{\mathbf{\theta}}_{1} = \tilde{\mathbf{\theta}} - \mathbf{J}(\tilde{\mathbf{\theta}},\vartheta^{\ast})^{-1} (\mathbf{u}(\mathbf{\theta}, \tilde{\mathbf{\theta}},\vartheta^{\ast}) - \tilde{\mathbf{h}}),
\end{equation}
where $\mathbf{J}(\tilde{\mathbf{\theta}},\vartheta^{\ast}) = \nabla_{\mathbf{\theta}} \mathbf{u}(\mathbf{\theta}, \tilde{\mathbf{\theta}},\vartheta^{\ast}) |_{\mathbf{\theta}=\tilde{\mathbf{\theta}}}$ is a Jacobian matrix. Similarly, we can define a \textit{two-step TML} (2TML) \textit{estimator} $\hat{\mathbf{\theta}}_{2}$ by replacing $\tilde{\mathbf{\theta}}$ with $\hat{\mathbf{\theta}}_{1}$, $\tilde{\mathbf{h}}$ with $\hat{\mathbf{h}}_{1}$, and $\mathbf{J}(\tilde{\mathbf{\theta}})$ with $\mathbf{J}(\hat{\mathbf{\theta}}_{1})$ in(\ref{equ1TML}).

In Theorem \ref{TeoSMAsymptotic1TML} of the Appendix we prove that, under general conditions including $n^{1/2}(\tilde{\mathbf{\theta}}-\mathbf{\theta}_{0})=O_{p}(1)$, the 1TML estimator satisfies the following asymptotic result:
\begin{equation*}
\sqrt{n}(\hat{\mathbf{\theta}}_{1}-\mathbf{\theta}_{0})\overset{\mathcal{L}}{\rightarrow}N(\mathbf{0},\mathbf{G}(\mathbf{\theta}_{0})^{-1}\mathbf{M}(\mathbf{\theta}_{0})\mathbf{G}(\mathbf{\theta}_{0})^{-\top}).
\end{equation*}
A similar result holds for the 2TML estimator. In the Monte Carlo simulations, reported in Section \ref{SecMonteCarlo}, we show that for finite sample sizes, the 1TML and the 2TML estimators are more efficient than the TQ$\tau$ estimator (with $10\%$ trimming) and have a reasonably robust behavior under outlier contamination. Therefore we propose, as a final procedure to estimate $\mathbf{\theta}_{0}$, the 2TML estimator starting with the TQ$\tau$ estimator with $10\%$ trimming. Numerical experiments show that further steps do not provide any significant improvement.

\section{The regression case}
\label{SecRegression} 
We now consider an AFT model for pairs of observations $(\mathbf{x}_{i},y_{i})$, $1 \leq i\leq n$, where $\mathbf{x}_{i}\in\mathbb{R}^{p}$ and $y_{i} \in\mathbb{R}$ satisfying
\begin{equation}
\label{model}y_{i}=\mu_{0}+\mathbf{\beta}_{0}^{\top}\mathbf{x}_{i} + \sigma_{0} u_{i}, \qquad i=1,\ldots,n ,
\end{equation}
where $y_{i}$ may represents a duration on the logarithmic scale and $\mathbf{x}_{i}$ the corresponding covariables vector. The slopes $\mathbf{\beta}_{0} \in\mathbb{R}^{p}$, the intercept $\mu_{0}$, and the scale $\sigma_{0}$ are unknown parameters. The errors $u_{i}$, $1\leq i\leq n$, are assumed to be i.i.d. and independent of $\mathbf{x}_{i}$. Moreover, the distribution of the carriers $\mathbf{x}_{i}$ is unknown. We assume that the error density is $f_{(0,1,\lambda_{0})}$, where $\lambda_{0}$ is an unknown shape parameter and $f_{\mathbf{\theta}}$ is given by (\ref{denggd}). We observe $(y_{i}^{\ast},\mathbf{x}_{i},\delta_{i})$, where $y_{i}^{\ast}=\min(y_{i},c_{i})$, $c_{1},\ldots,c_{n}$ are i.i.d. censoring times, which are independent of the $u_{i}$'s. We put $\delta_{i}=1$ if $y_{i}^{\ast}=y_{i}$ and $\delta_{i}=0$ otherwise. We write $\mathbf{\gamma}_{0}=(\mathbf{\theta}_{0},\mathbf{\beta}_{0})$, where $\mathbf{\theta}_{0}=(\mu_{0},\sigma_{0},\lambda_{0})$.

Let $\mathbf{\gamma}=(\mathbf{\theta},\mathbf{\beta})$ and $u=(y-\mu-\mathbf{\beta}^{\top}\mathbf{x})/\sigma$. Let $\mathbf{d}(\mathbf{x},y,\mathbf{\gamma})$ and $\mathbf{s}(\mathbf{x},y,\mathbf{\gamma})$ denote the $(3+p)$-component non-censored and censored regression score function vectors respectively. The first $3$ components of $\mathbf{d}(\mathbf{x},y,\mathbf{\gamma})$ and $\mathbf{s}(\mathbf{x},y,\mathbf{\gamma})$ are given by the right-hand sides of (\ref{zscore1})-(\ref{zscore3}) and (\ref{sscore1})-(\ref{sscore3}) respectively. In addition, for $k=1,\ldots,p$, we have
\begin{align*}
d_{3+k}(\mathbf{x},y,\mathbf{\gamma})  &  =\frac{x_{k}}{\sigma}\xi_{\lambda}(u),\\
s_{3+k}(\mathbf{x},y,\mathbf{\gamma})  &  =-\frac{x_{k}}{\sigma}\frac{f_{\lambda}(u)}{S_{\lambda}(u)}.
\end{align*}
Then, simple derivations shows that the ML estimator of $\mathbf{\gamma}$ is the solution of
\begin{equation} \label{equML3}
\frac{1}{n} \sum_{j=1}^{n}\mathbf{v}(\mathbf{x}_{j},y_{j}^{\ast},\delta_{j},\mathbf{\gamma})=\mathbf{0},
\end{equation}
where $\mathbf{v}(\mathbf{x},y,\delta,\mathbf{\gamma})=\delta\mathbf{d}(\mathbf{x},y,\mathbf{\gamma})+(1-\delta)\mathbf{s}(\mathbf{x},y,\mathbf{\gamma})$. A similar expression as (\ref{equML2}) can be obtained, where the \textit{semiparametric cdf} is defined by
\begin{align*}
H_{n,\mathbf{\gamma}}(z,\mathbf{z})  &  = \frac{1}{n} \sum_{i=1}^{n} E_{\mathbf{\gamma}} \left[  \left.  I\left(  y \leq z \right)  \right\vert y_{i}^{\ast}, \mathbf{x}_{i}, \delta_{i} \right]  I(\mathbf{x}_{i} \leq\mathbf{z})\\
&  =\frac{1}{n}\sum\delta_{i} I\left(  y_{i} \leq z\right)  I(\mathbf{x}_{i}\leq\mathbf{z})\\
&  + \frac{1}{n} \sum\left(  1-\delta_{i}\right)  \frac{\left[ F_{\mathbf{\gamma}}(z) - F_{\mathbf{\gamma}}(y_{i}^{\ast}) \right]  ^{+}}{1-F_{\mathbf{\gamma}}(y_{i}^{\ast})} I(\mathbf{x}_{i}\leq\mathbf{z}).\label{Hempreg}
\end{align*}
We denote by $E_{n,\mathbf{\gamma}}$ the expectation with respect to $H_{n,\mathbf{\gamma}}$. In the next two subsections we define a robust and efficient procedure for estimating $\mathbf{\gamma}$. In a first step an initial highly robust but not necessarily efficient estimator of $\mathbf{\gamma}_{0}$ is computed; the second step uses a highly robust and asymptotically efficient estimator of $\mathbf{\gamma}_{0}$ based on a 1TML procedure.

\subsection{The initial regression estimator}
\label{SecMMTQtaureg} 
An initial estimator of $\mathbf{\gamma}_{0}$ can be defined as follows:
\begin{enumerate}
\item Let $\bar{\mu}$ and $\bar{\mathbf{\beta}}$ be MM-estimators for censored data of $\mu_{0}$ and $\mathbf{\beta}_{0}$ as defined in \citet{salibianbarrera2008}.

\item Let $\omega_{j}=y_{j}-\bar{\mathbf{\beta}}^{\top}\mathbf{x}_{j}$ and $\omega_{j}^{\ast}=y_{j}^{\ast}-\bar{\mathbf{\beta}}^{\top}\mathbf{x}_{j}=\min(\omega_{j},c_{j}^{\ast})$ where $c_{j}^{\ast}=c_{j}-\bar{\mathbf{\beta}}^{\top}\mathbf{x}_{j}$. For large $n$, the distribution of the $\omega_{j}$s is close to $F_{\mathbf{\theta}_{0}}$ for $1\leq j\leq n$. We can then estimate $\mathbf{\theta}_{0}$ using the observations $(\omega_{1}^{\ast},\delta_{1}),\ldots,(\omega_{n}^{\ast},\delta_{n})$ by means of a TQ$\tau$ estimator that will be denoted by $\tilde{\mathbf{\theta}}=(\tilde{\mu},\tilde{\sigma},\tilde{\lambda})$.

\item The initial regression estimator of $\mathbf{\gamma}_{0}$ is $\tilde{\mathbf{\gamma}} = (\tilde{\mathbf{\theta}},\bar{\mathbf{\beta}})$. It will be called \textit{MM-TQ$\tau$ estimator}.
\end{enumerate}

In Theorem \ref{TeoSMConsistencyTQtaureg} we show that, under general conditions, if $\bar{\mathbf{\beta}}$ is $n^{1/2}$ consistent, that is if
\begin{equation}
\label{betab}n^{1/2}(\bar{\mathbf{\beta}}-\mathbf{\beta}_{0}) = O_{p}(1),
\end{equation}
then $\tilde{\mathbf{\theta}}$ satisfies
\begin{equation}
\label{thetab}n^{1/2}(\tilde{\mathbf{\theta}}-\mathbf{\theta}_{0}) = O_{p}(1).
\end{equation}
The result (\ref{betab}) remains still a conjecture, however \citet{salibianbarrera2008} provide compelling arguments in favor of this conjecture (see in particular their Theorem 6). Besides, their Monte Carlo study seems to confirm this conjecture.

\subsection{The final regression estimator}
\label{Sec1TMLreg} 
Let $\mu(\mathbf{x})=\mu+\mathbf{x}^{\top}\mathbf{\beta}$ and $\tilde{\mu}(\mathbf{x})=\tilde{\mu}+\mathbf{x}^{\top}\bar{\mathbf{\beta}}$. Then, the standardized residuals with respect to the initial estimator are $\tilde{r}_{i}^{\ast}=(y_{i}^{\ast}-\tilde{\mu}(\mathbf{x}_{i}))/\tilde{\sigma}$ and can be used to obtain the cutoff point $\vartheta^{\ast}$ and the weights
\begin{equation*}
w(\mathbf{x},y,\mathbf{\gamma},\vartheta)=w\left(  \frac{y-\tilde{\mu }(\mathbf{x})}{\tilde{\sigma}},\tilde{\lambda},\vartheta\right) .
\end{equation*}
The \textit{TML regression estimator} $\hat{\mathbf{\gamma}}$ is the solution of the equations
\begin{equation}
\mathbf{u}(\mathbf{\gamma},\tilde{\mathbf{\gamma}},\vartheta^{\ast})=\tilde{\mathbf{h}}, \label{equTMLreg}
\end{equation}
where
\begin{equation*}
\mathbf{u}(\mathbf{\gamma},\tilde{\mathbf{\gamma}},\vartheta^{\ast})=E_{n,\tilde{\mathbf{\gamma}}}\left[  w(\mathbf{x},y,\tilde{\mathbf{\gamma}},\vartheta^{\ast})\mathbf{d}(\mathbf{x},y,{\mathbf{\gamma}})\right] ,
\end{equation*}
and $\tilde{\mathbf{h}}=(\tilde{h}_{1},\tilde{h}_{2},\tilde{h}_{3},\tilde{\mathbf{h}}_{4}^{\top})^{\top}$, where $\tilde{h}_{1},\tilde{h}_{2},\tilde{h}_{3}$ are defined in (\ref{htilde}) and
\begin{equation*}
\tilde{\mathbf{h}}_{4}=\frac{1}{n}\sum_{i=1}^{n}E_{\tilde{\lambda}}\left[ w\left(  u,\tilde{\lambda},\vartheta^{\ast}\right)  \xi_{\tilde{\lambda}}(u)\mathbf{x}_{i}\right] .
\end{equation*}
The \textit{1TML regression estimator} is obtained by applying one Newton-Raphson iteration to equations (\ref{equTMLreg}). This estimator turns out to be
\begin{equation*}
\hat{\mathbf{\gamma}}_{1}=\tilde{\mathbf{\gamma}}-\mathbf{J}(\tilde{\mathbf{\gamma}},\vartheta^{\ast})^{-1}(\mathbf{u}(\mathbf{\gamma},\tilde{\mathbf{\gamma}},\vartheta^{\ast})-\tilde{\mathbf{h}}),
\end{equation*}
where $\mathbf{J}(\tilde{\mathbf{\gamma}},\vartheta^{\ast})=\nabla_{\mathbf{\gamma}}\mathbf{u}(\mathbf{\gamma},\tilde{\mathbf{\gamma}},\vartheta^{\ast})|_{\mathbf{\gamma}=\tilde{\mathbf{\gamma}}}$ is a Jacobian matrix. The \textit{2TML regression estimator} is defined in the obvious way.

In Theorem \ref{TeoSMAsymptotic1TMLreg} of the Appendix we prove that, under general conditions including $n^{1/2}(\tilde{\mathbf{\gamma}}-\mathbf{\gamma}_{0})=O_{p}(1)$, the 1TML regression estimator satisfies the following asymptotic result:
\begin{equation*}
\sqrt{n}(\hat{\mathbf{\gamma}}_{1} - \mathbf{\gamma}_{0}) \overset {\mathcal{L}}{\rightarrow} N(\mathbf{0},\mathbf{G}(\mathbf{\gamma}_{0})^{-1} \mathbf{M}(\mathbf{\gamma}_{0})\mathbf{G}(\mathbf{\gamma}_{0})^{-\top}),
\end{equation*}
where $\mathbf{M}(\mathbf{\gamma}) = E\left(  \mathbf{v}(\mathbf{x},y,\delta,\mathbf{\gamma}) \mathbf{v}(\mathbf{x}, y,\delta,\mathbf{\gamma})^{\top}\right)  $ and $\mathbf{G}(\mathbf{\gamma}) = E\left(  \nabla_{\mathbf{\gamma}} \mathbf{v}(\mathbf{x}, y,\delta,\mathbf{\gamma}) \right)$. The same result can be obtained for the 2TML regression estimator.

\section{Monte Carlo experiments}
\label{SecMonteCarlo}

\subsection{The case without covariates}
A first set of experiments was run in the case without covariates. In general, the results are similar to those described in \citet{agostinelli2014a} for the non censored case and we report here only the most representative cases. We compared the following estimators: ML, TQ$\tau$, 1TML, 2TML, and the one step WL estimator (1SWL) as defined in Section \ref{SecSM1SWL} of the Appendix.

The TQ$\tau$ estimator was defined using $\rho_{1}$ and $\rho_{2}$ in the Tukey's bisquare family (\ref{equTukeyRho}) with $c$ equal to $1.548$ and $6.08$ respectively, and $b=0.5$. The values $1.548$ and $b=0.5$ have been chosen so that the regression estimator based on the $\tau$-scale has a finite sample breakdown point equal to $0.5$. The value $6.08$ makes the asymptotic efficiency equal to $0.95$ in the case of normal errors. To compute the 1SWL estimator, we used a normal kernel with bandwidth $h=0.3\widehat{\sigma}$ in all experiments, where the scale estimator $\widehat{\sigma}$ is the scale initial estimate.

To compare the global performances of the different estimators we consider, for each set of parameters $\mathbf{\theta}=(\mu, \sigma, \lambda)$, the total variation distance (TVD)
\begin{equation*}
\text{TVD}(\mathbf{\theta})=\int\left\vert f_{\mathbf{\theta}}(y)-f_{\mathbf{\theta}_{0}}(y)\right\vert dy
\end{equation*}
between the density $f_{\mathbf{\theta}}$ and the true underlying density $f_{\mathbf{\theta}_{0}}$. The performance of the estimator $\mathbf{\hat {\theta}}$ is measured by the mean value of TVD($\mathbf{\hat{\theta}}$), that is by
\begin{equation*}
\text{MTVD}(\mathbf{\hat{\theta}})=E(\text{TVD}(\mathbf{\hat{\theta}})).
\end{equation*}
MTVD$(\mathbf{\hat{\theta}})$ clearly measures the quality of the estimated density. It is estimated, using the simulated values $\mathbf{\hat{\theta}}_{k}$ ($1\leq k\leq N$) of the estimator, by
\begin{equation*}
\text{ATVD}(\mathbf{\hat{\theta}})=\frac{1}{N}\sum_{k=1}^{N}\text{TVD}%
(\mathbf{\hat{\theta}}_{k}).
\end{equation*}

\subsubsection{Simulation under the nominal model}
We studied the efficiency of the estimators under the nominal model for $n=50$, $100$, $400$ and $1000$ and for $\lambda_{0}$ equal $1$ and $2$. Without loss of generality we took $\mu_{0}=0$ and $\sigma_{0}=1$. We considered two censoring fractions: $15\%$ and $25\%$. The number of replications was $1200$. Figure \ref{FigWithout085} (top) reports the ATVD of the robust estimators divided by the ATVD of the ML estimator as a function of the sample size for $\lambda=1$ and censoring fraction of $15\%$. This ratio can be interpreted as a measure of relative efficiency that we call \textit{TVD efficiency}. As expected, the TVD efficiency of 1TML, 2TML and 1SWML is markedly larger than the efficiency of the initial estimator as the sample size grows. Moreover, when $n$ increases the TVD efficiency becomes close to 1, that is, it becomes close to the efficiency of the ML estimator. Similar patterns are observed in the other simulated cases.

\subsubsection{Simulation under point mass contamination}
In a second Monte Carlo experiment, we compared TQ$\tau$ with $10\%$ trimming, 1TML, 2TML and 1SWL under point mass contamination for $n=50$, $100$, $400$, $1000$, $\lambda_{0}=1$, $\lambda_{0}=2$ and two censoring fractions: $15\%$ and $25\%$. We generated $90\%$ ``good'' observations $y_{j}$ according to the GLG model and $10\%$ ``outliers'' at the point $y_{0}$. We then varied the value of $y_{0}$ from $-10$ to $10$ with a step of $0.5$. This kind of point mass contamination is generally the least favorable one and allows evaluation of the maximal bias an estimator can incur. For each value of $y_{0}$, the number of replications was $1200$. Figure \ref{FigWithout085} (bottom) reports the average TVD (ATVD) of the estimated densities as a function of $y_{0}$. The results show that the ATVD of TQ$\tau$, 1TML and 2TML are comparable and very stable over the whole range of $y_{0}$ values, while 1SWL provides a very high ATVD for positive values of $y_{0}$. The ATVD of ML is not shown because of its unbounded behavior.

\begin{figure}[ptb]
\begin{center}
\includegraphics[width=0.8\textwidth]{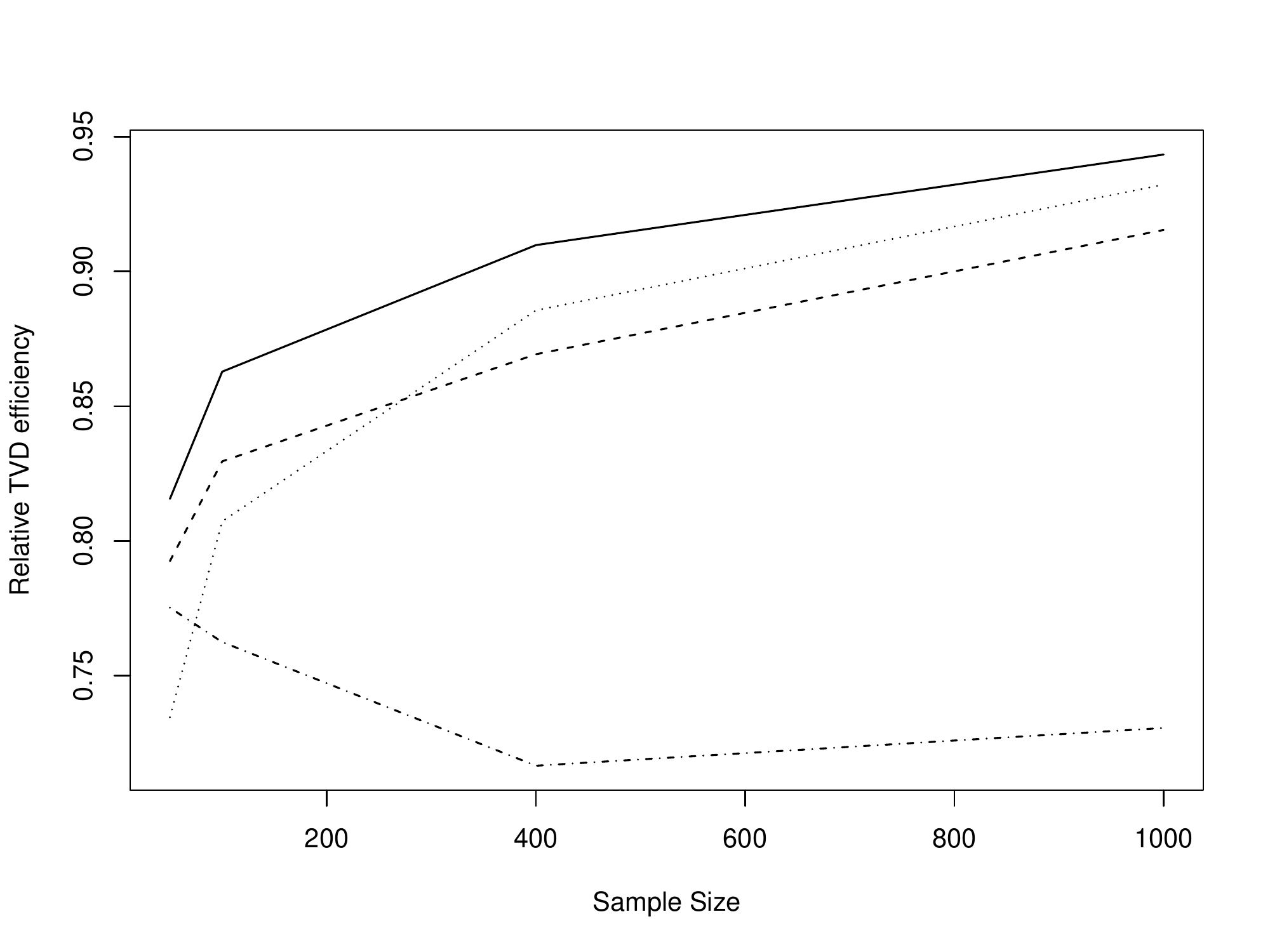}
\includegraphics[width=0.8\textwidth]{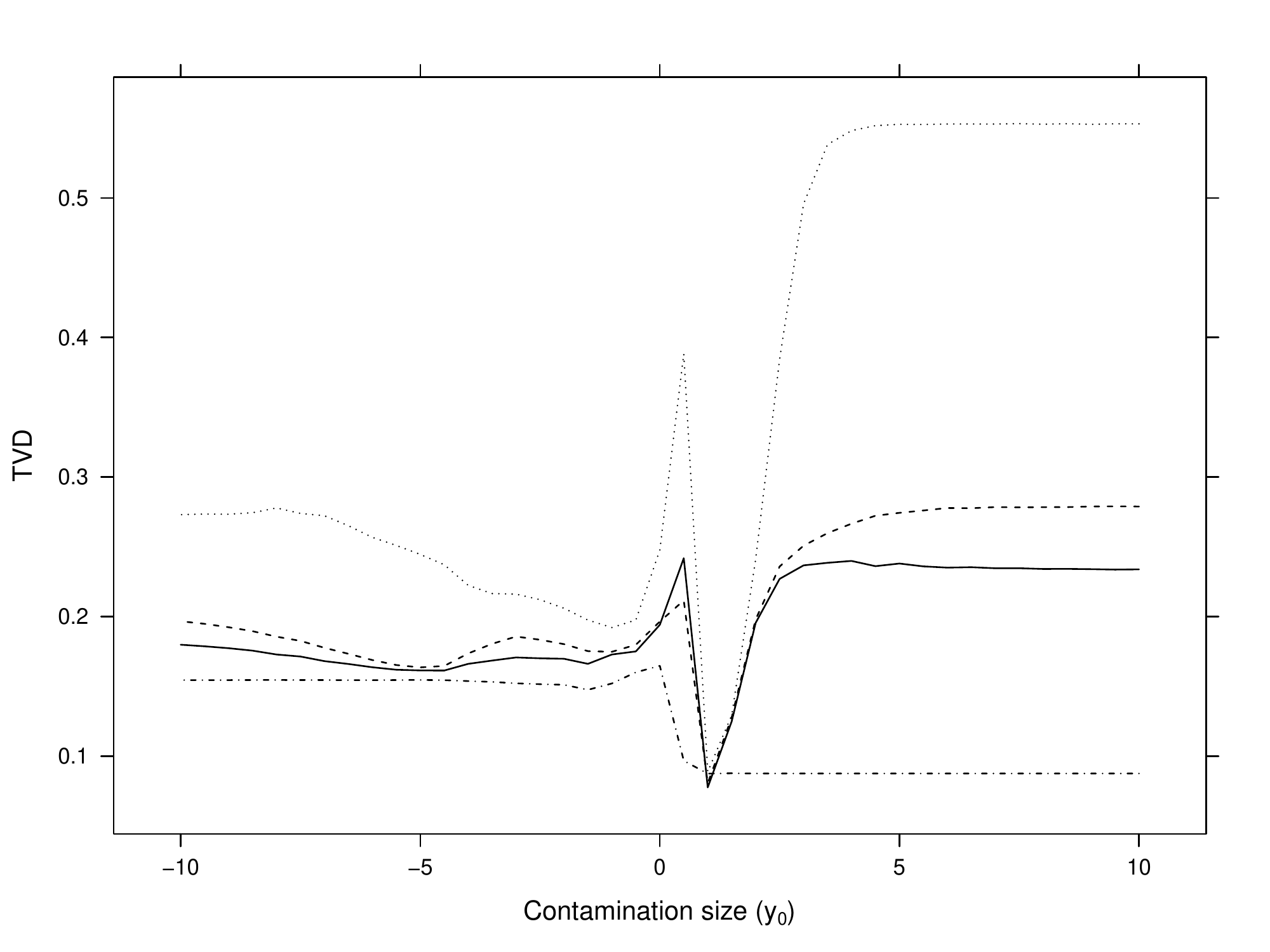}
\end{center}
\caption{Case without regressors. Top: Estimated relative TVD efficiency versus sample size. Bottom:
Average TVD under $10\%$ contamination and sample size $n=100$. Legend: TQ$\tau$ (dashed and dotted line), 1TML (dashed line), 2TML (solid line) and 1SWL (dotted line). The censoring proportion is $15\%$ and $\lambda=1$.}
\label{FigWithout085}
\end{figure}

\subsection{Regression case}
A second set of experiments was run to investigate the behavior of the MM-TQ$\tau$ estimator defined in Section \ref{SecMMTQtaureg}, the 1TML and 2TML estimators defined in Section \ref{Sec1TMLreg}, and the 1SWL estimator defined in Section \ref{SecSM1SWLreg} of the Appendix.

In each experiment, $n$ pairs of observations $(\mathbf{x}_{i},y_{i}^{\ast})$ were generated according to the regression model
\begin{align}
y_{i} &  =\mu_{0}+\mathbf{\beta}_{0}^{\top}\mathbf{x}_{i}+\sigma_{0}u_{i},\label{simmod1}\\
y_{i}^{\ast} &  =\min(y_{i},c_{i}),\label{simmod2}\\
c_{i} &  =\mu_{c}+e_{i},\label{simmod3}\\
\delta_{i} &  =1\text{ if }y_{i}^{\ast}=y_{i}\text{ and }0\text{ otherwise,}\label{simmod4}
\end{align}
with $\mu_{0}=0$, $\sigma_{0}=1$, $\lambda_{0}=1$ and $2$, $\mathbf{\beta}_{0}^{\top}=(2,3)$, $\mathbf{x}_{i}=(x_{i1},x_{i2})$, where $x_{i1}\sim N(0,1)$, $x_{i2}\sim\mathcal{B}(1,0.4)$ with $x_{i1}$ and $x_{i2}$ independent and $e_{i}\sim GLG(\mu_{c},1,1)$. The parameter $\mu_{0}$ was chosen so that the censoring fraction was $15\%$ or $25\%$.

To compare the estimators we defined, for each set of parameters $(\mathbf{\theta}, \mathbf{\beta})$ and covariable vector $\mathbf{x}$,
\begin{equation*}
\text{TVD}(\mathbf{\theta}, \mathbf{\beta}, \mathbf{x}) = \int\left\vert f_{(\mathbf{\theta},\mathbf{\beta})}(y|\mathbf{x})-f_{(\mathbf{\theta}_{0},\mathbf{\beta}_{0})}(y|\mathbf{x})\right\vert dy ,
\end{equation*}
where
\begin{equation*}
f_{(\mathbf{\theta}, \mathbf{\beta})}(y|\mathbf{x})=f_{\lambda}\left( (y -\mu- \mathbf{\beta}^{\top}\mathbf{x})/\sigma \right) .
\end{equation*}
Then, we measured the performance of the estimator $(\mathbf{\hat{\theta}},\mathbf{\hat{\beta}})$ by the mean value of $\text{TVD}(\mathbf{\hat{\theta}},\mathbf{\hat{\beta}},\mathbf{x})$, where $\mathbf{x}$ is independent of the sample used to compute $(\mathbf{\hat{\theta}}, \mathbf{\hat{\beta}})$, that is by
\begin{equation*}
\text{MTVD}(\mathbf{\hat{\theta}},\mathbf{\hat{\beta}})=E(\text{TVD}(\mathbf{\hat{\theta}},\mathbf{\hat{\beta}},x).
\end{equation*}
The expectation was taken on $(\mathbf{\hat{\theta}},\mathbf{\hat{\beta}})$ and $\mathbf{x}$ and estimated using the simulated values $(\mathbf{\hat{\theta}}_{k},\mathbf{\hat{\beta}}_{k})$ and $\mathbf{x}_{k1}^{\ast},\ldots,\mathbf{x}_{ki}^{\ast},\ldots,\mathbf{x}_{kn}^{\ast}$ ($1\leq k \leq N$):
\begin{equation*}
\text{ATVD}(\mathbf{\hat{\theta},\hat{\beta}})=\frac{1}{N n} \sum_{k=1}^{N}\sum_{i=1}^{n}\text{TVD}(\mathbf{\hat{\theta}}_{k},\mathbf{\hat{\beta}}_{k},\mathbf{x}_{ki}^{\ast}).
\end{equation*}

\subsubsection{Simulation under the nominal model}
We studied the efficiency of the estimators under the nominal model for $n=50$, $100$, $400$ and $1000$. Figure \ref{FigWith085} (top) shows the relative TVD efficiency of MM-TQ$\tau$, 1TML, 2TML and 1SWL with respect to ML for $\lambda_{0}=1$, $\mu_{0}=0$, $\sigma_{0}=1$, and censoring fraction $0.15\%$.

The efficiency of 1TML, 2TML and 1SWL is clearly higher than the efficiency of MM-TQ$\tau$. This is an expected result, since 1TML, 2TML and 1SWL are asymptotically fully efficient under the nominal model.

\subsubsection{Simulation under point mass contamination}
We also studied the behavior of the estimators under point mass contamination for $n=50$, $100$, $400$, $1000$. The values of the parameters were the same as in the case of no contamination. We generated $n$ ``good'' observations $(\mathbf{x}_{i},y_{i}^{\ast})$ according to (\ref{simmod1})-(\ref{simmod4}). We then replaced $10\%$ values $y_{i}^{\ast}$ with a value $y_{0}$ ranging from $-10$ to $20$. For each value of $y_{0}$ the number of replications was $1200$. Figure \ref{FigWith085} (bottom) shows ATVD of the estimators as a function of $y_{0}$ for sample size $n=100$ and contamination level $10\%$. We observe that MM-TQ$\tau$, 1TML, and 2TML are very resistant under point mass contamination, while the 1SWL is highly sensitive to outliers on the right tail of the distribution.

\subsubsection{Empirical finite sample breakdown point}
We were not able to obtain the breakdown point of the proposed estimators. To fill this gap, we performed a Monte Carlo simulation to explore the behavior of the maximum mean square error as a function of the contamination level. This provides information about the highest contamination the proposed estimators could cope with and hence about the finite sample breakdown point. We used the same setting as in the previous subsection, $n=1000$, $\lambda_{0}=1$, and censoring fraction $0.15$. Several values of the contamination level $\epsilon$ in the interval $[0,0.3]$ and several values $y_{0}$ in the interval $[-100,100]$ were considered. For each pair $(\epsilon,y_{0})$, we run $100$ Monte Carlo replications and computed the maximum MSE (MMSE) of the regression parameters $(\mu,\beta)$ --note that $\mu$ is the intercept--, the scale parameter ($\sigma$) and the shape parameter ($\lambda$). Results for the regression parameters (slopes) are reported in Figure \ref{FigBreakdown}; they show that the MMSE starts to increase rapidly around the $20\%$ level. This behavior is consistent with the MMSE of the initial regression parameters provided by the MM non parametric estimator.

\begin{figure}[ptb]
\begin{center}
\includegraphics[width=0.8\textwidth]{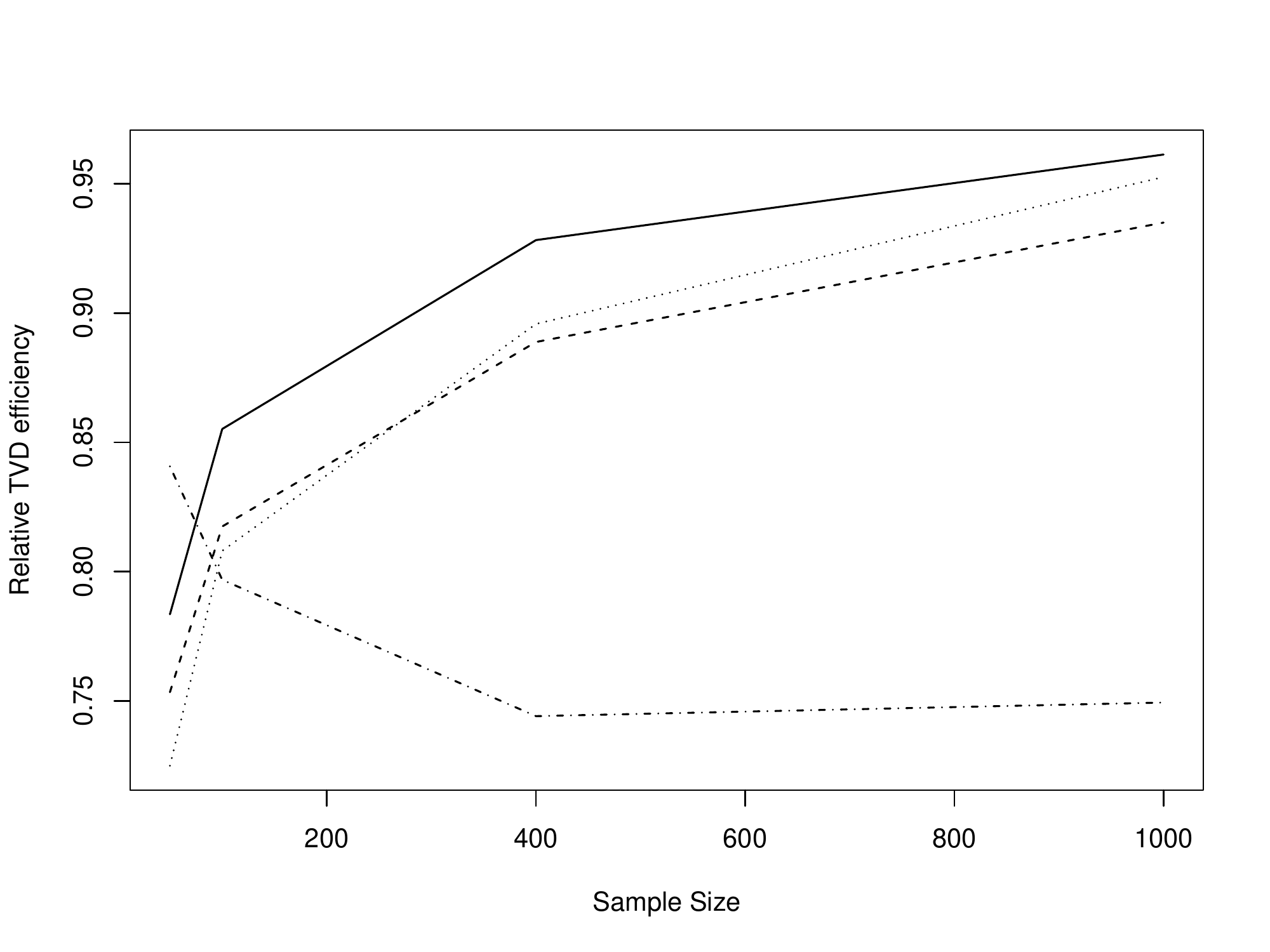}
\includegraphics[width=0.8\textwidth]{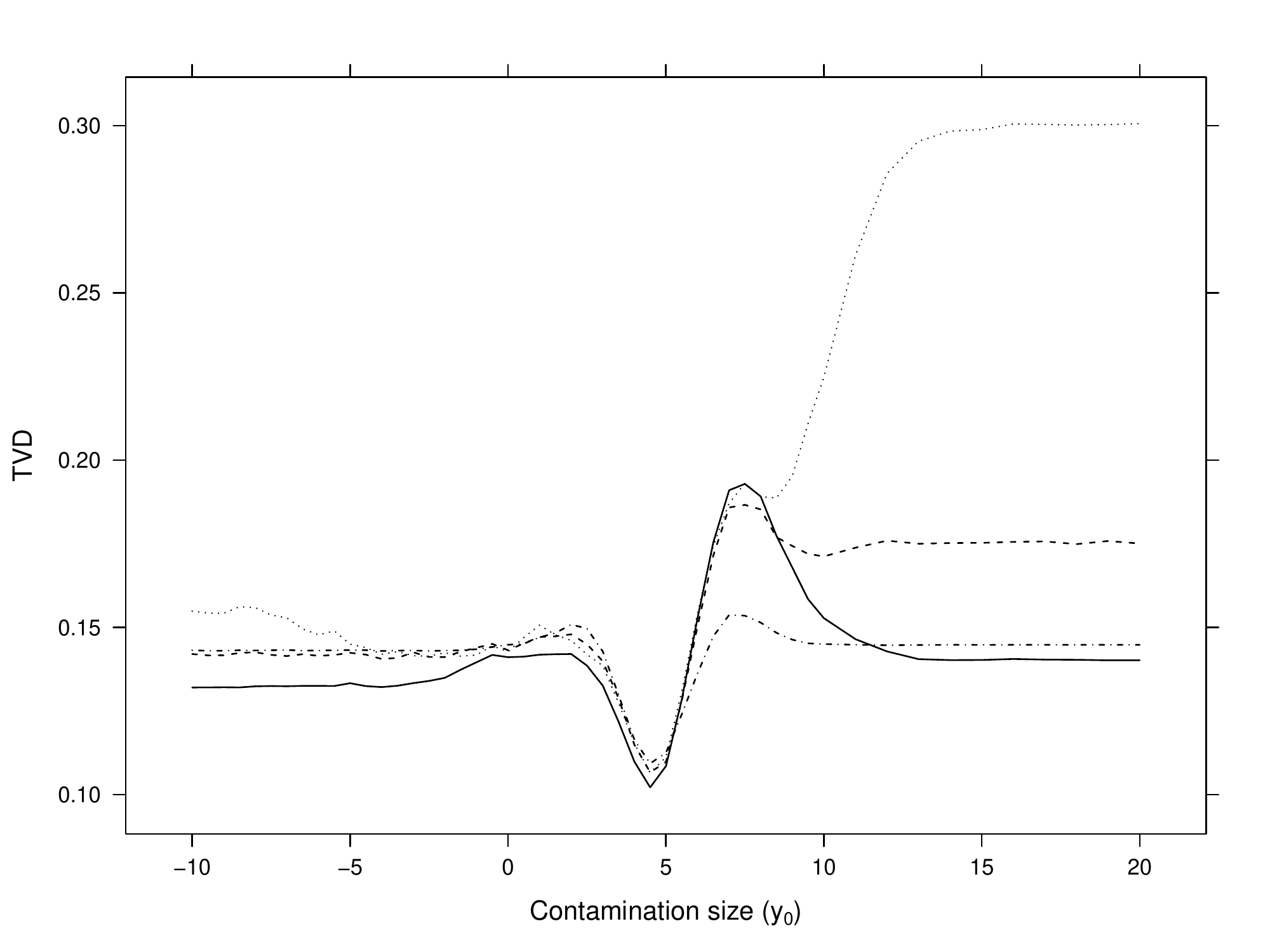}
\end{center}
\caption{Regression case. Top: Estimated relative TVD efficiency versus sample size. Bottom: Average TVD under $10\%$ contamination and sample size $n=100$. Legend: MM-TQ$\tau$ (dashed and dotted line), 1TML (dashed line), 2TML (solid line) and 1SWL (dotted line). The censoring proportion is $15\%$ and $\lambda=1$.}
\label{FigWith085}
\end{figure}

\begin{figure}[ptb]
\begin{center}
\includegraphics[width=0.9\textwidth]{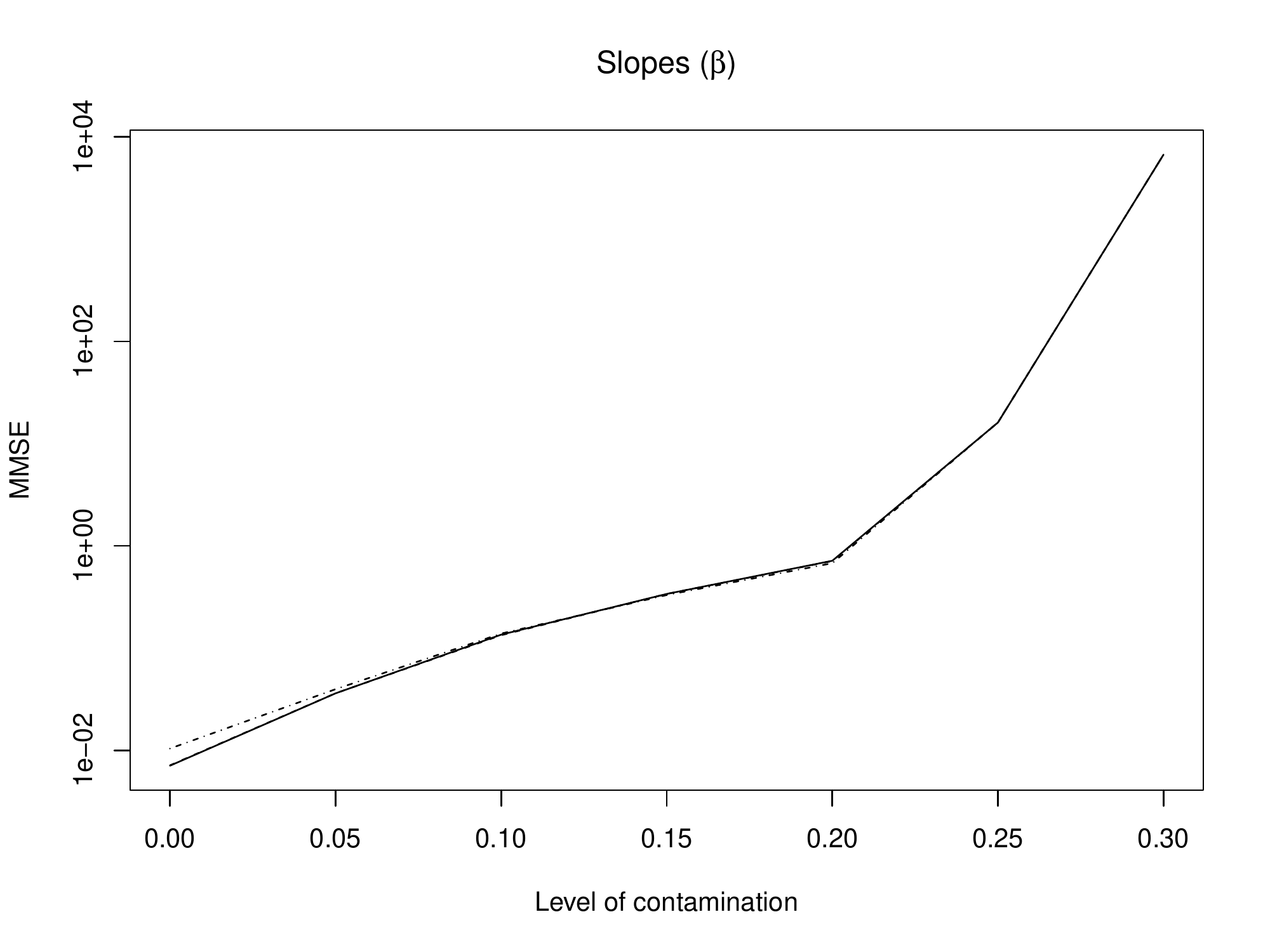}
\end{center}
\caption{Estimated maximum MSE of the regression parameters (slopes) for MM-TQ$_\tau$ (dashed and dotted line), 1TML (dashed line), and 2TML (solid line). The censoring proportion is $15\%$, $n=1000$ and $\lambda=1$. The MMSEs of the 1TML and 2TML are almost always overlapping. The y-axis is on log scale.}
\label{FigBreakdown}
\end{figure}

\begin{figure}[ptb]
\begin{center}
\vspace{-15mm}
\includegraphics[width=0.58\textwidth]{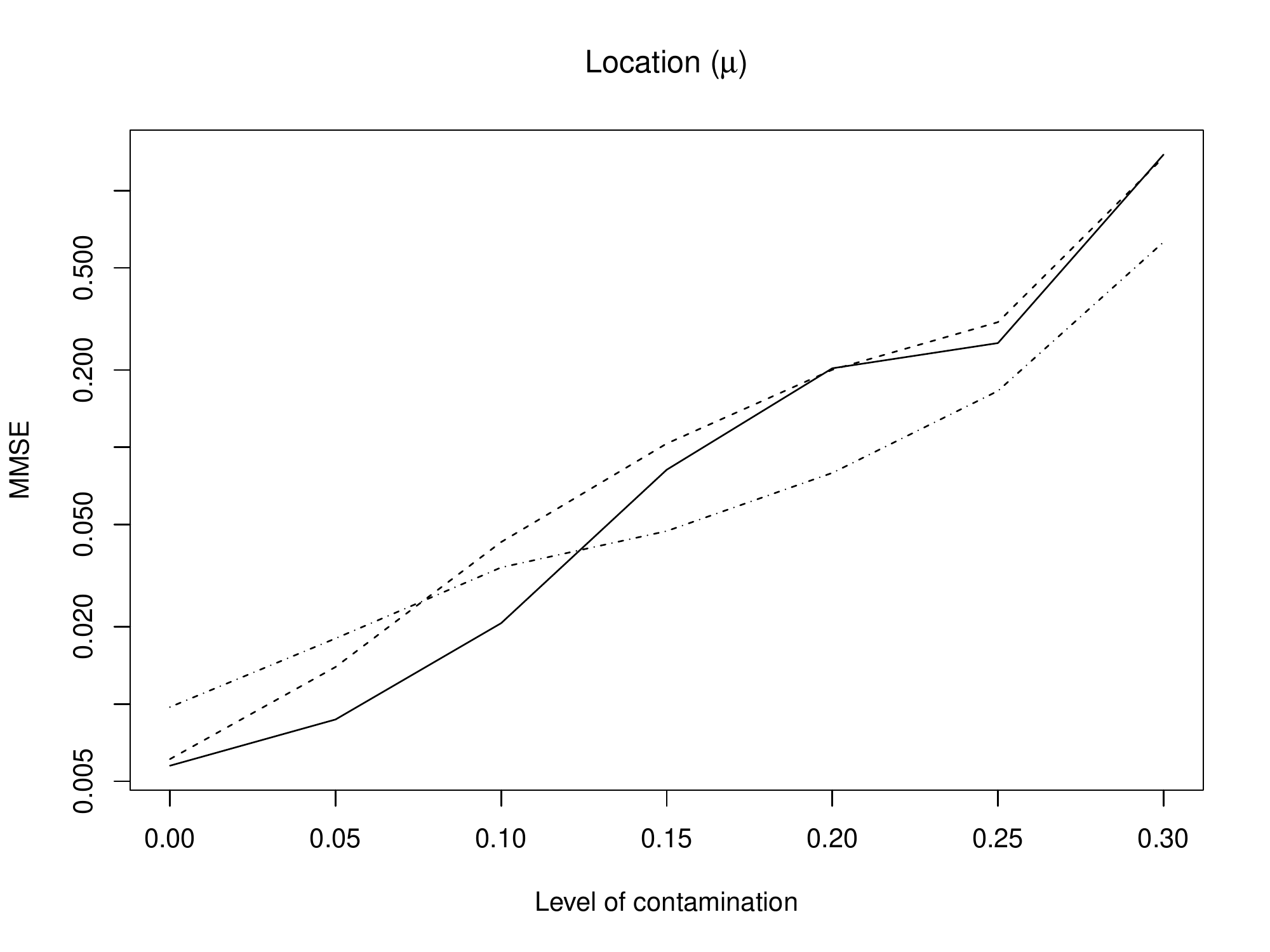}
\includegraphics[width=0.58\textwidth]{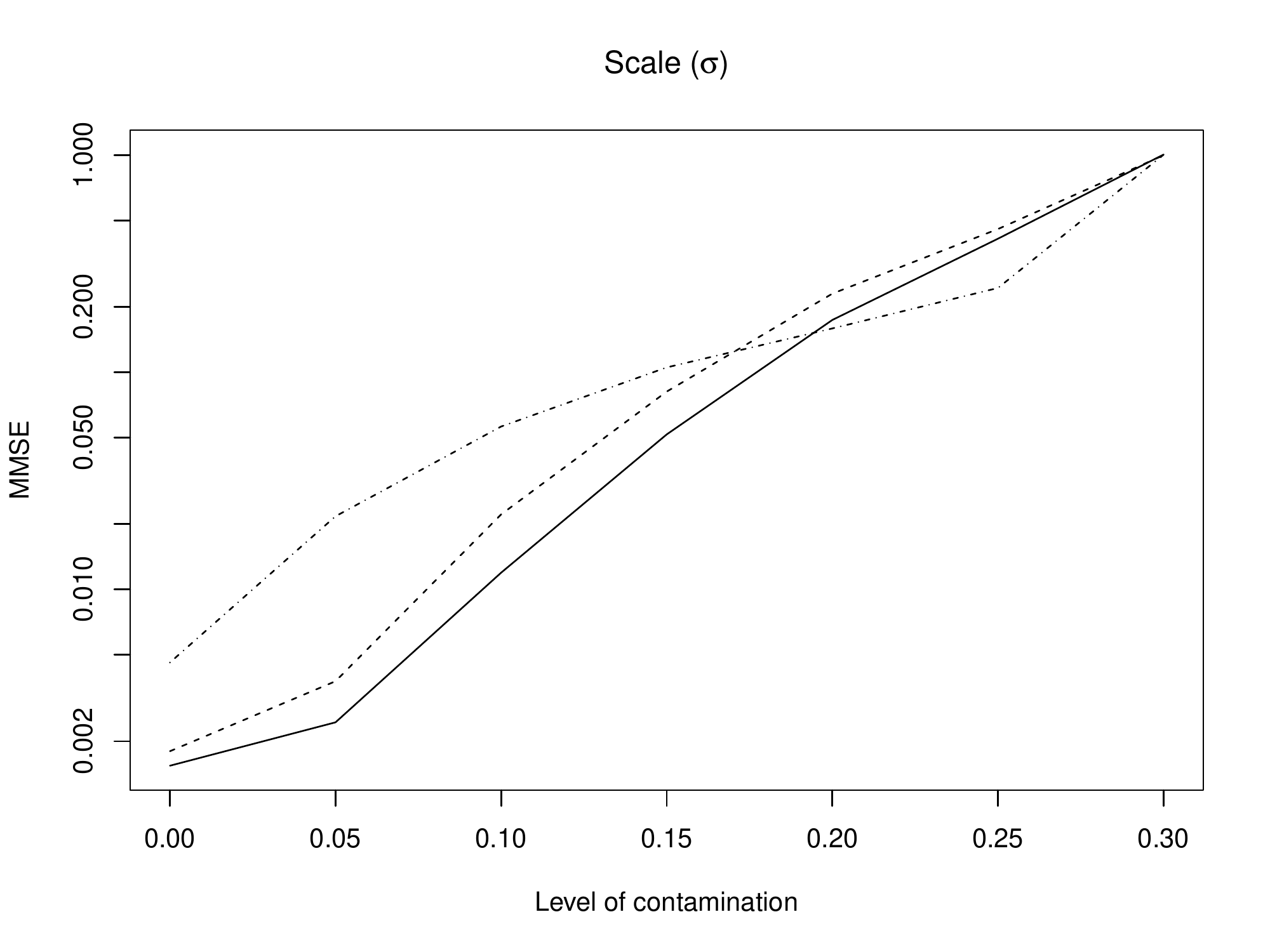}
\includegraphics[width=0.58\textwidth]{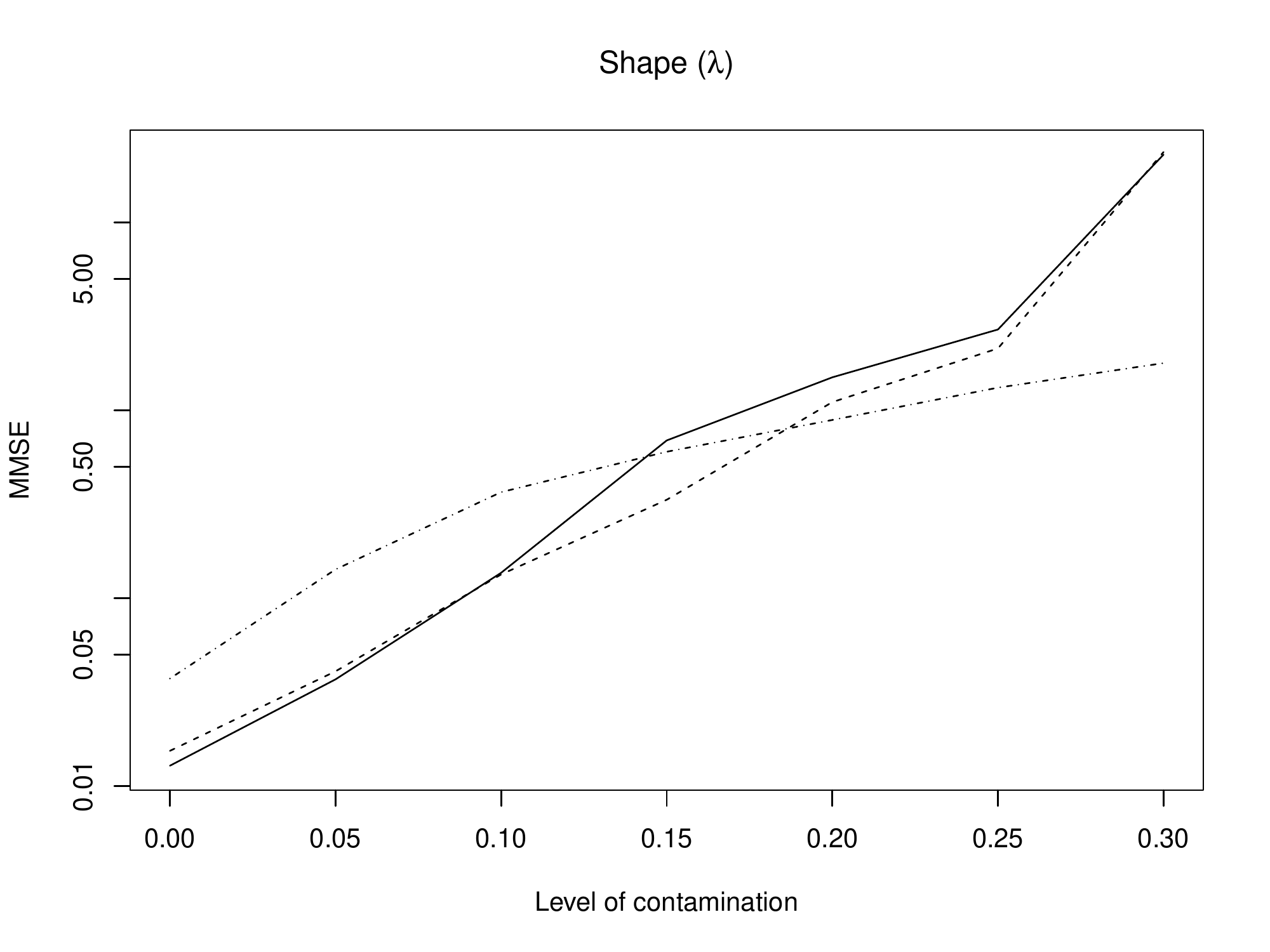}
\end{center}
\caption{Regression case. Estimated maximum MSE as function of contamination level for TQ$\tau$ (dashed and dotted line), 1TML (dashed line) and 2TML (solid line). Top: location parameter, middle: scale parameter, bottom: shape parameter. The censoring proportion is $15\%$, $n=1000$ and $\lambda=1$.}
\label{FigwithLocationScaleShape}
\end{figure}

\section{Illustrations with real data}
\label{SecExamples} 
In modern hospital management, stays are classified into ``diagnosis related groups'' (DRGs; \citet{fetter1980} which are designed to be as homogeneous as possible with respect to diagnosis, treatment, and resource consumption. The ``mean cost of stay'' of each DRGs is periodically estimated with the help of administrative data on a national basis and used to determine ``standard prices'' for hospital funding and reimbursement. Since it is difficult to measure cost, ``length of stay'' (LOS) is often used as a proxy. In designing and ``refining'' the groups, the relationship between LOS and other variables which are usually available on administrative files has to be assessed and taken into account. We discuss two examples in this domain.

\subsection{Major cardiovascular interventions}
In a first example, we consider a sample of $75$ stays in a particular
hospital and DRG ``Major cardiovascular interventions''. Of these stays, $45$ were censored because the patients were transferred to a different hospital before dismissal. The data -- shown in Figure \ref{FigExaMCI} and made available in \citet{marazzi2013} -- were first analyzed in
\citet{locatelli2010} and \citet{locatelli2013}. These authors studied the relationship between LOS and two covariates: Age of the patient ($x_{1}$) and Admission type ($x_{2}=0$ for planned admissions, $x_{2}=1$ for emergency admissions) with the help of the model $y=\alpha+\beta_{1}x_{1}+\beta_{2}x_{2}+\gamma x_{1}x_{2}+\sigma u$, where $y=\log($LOS$)$ and $u$ is following a Gaussian model. They observed that two young patients had exceptionally high non censored LOS and, as a consequence, the ML estimator yielded an unexpected large estimate of the interaction $\gamma$. Therefore, they proposed the use of a robust parametric procedure called ``weighted maximum likelihood'' (WML/G) based on the Gaussian error model (that performed better than log-Weibull). They compared WML/G with other published robust procedures an found that the robust estimates of $\gamma$ were close to zero.

Here, we assume a GLG error model and consider the ML, 1TML, 2TML, 1SWL and 2SWL regression estimates reported in Table \ref{TabExaMCI}. Estimated regression lines are reported in Figure \ref{FigExaMCI} as well. For comparison, we also report the ML estimate with Gaussian errors (ML/G) and WML/Gauss. We first notice the good agreement among the robust coefficient estimates based on GLG. The one and two steps TML and WL estimates provide the same inferences as ML after removal of the outliers.

Apart from $\mu$, $\sigma$ and $\beta_{2}$, ML yields larger absolute values for the estimates; however, none of them are significant because standard errors are inflated by the outliers. Not surprisingly, the main differences between the robust procedures based on GLG and those based on the Gaussian model concern the intercept terms ($\mu$ and $\beta_{1}$); however, the robust prediction lines based on GLG (Figure \ref{FigExaMCI}) seem to provide a better fit to the bulk of the data. This observation is supported by the plots in Figure \ref{FigExaMCIFhat}, where three types of distributions of the standardized residuals are displayed: KM, parametric (normal and GLG), and semiparametric (normal and GLG). Note the very large steps of KM corresponding to the two extreme non censored observations. The reason is that KM puts the mass of several censored residuals on these two points. The ML survival functions are strongly affected by these two points (see Figure 4a in \citet[]{locatelli2010}). Both WML/Gauss and 2TML distribution functions behave much better: with two exceptions, their residuals follow the models very well. However 2TML is clearly better in the left tail. Finally, we note that the use of GLG provides a reasonable fit for the two young patients with high non censored LOS, while the censored observations corresponding to emergency admission in the right bottom corner are considered outliers.

\begin{table}[ptb]
\caption{Estimates of the regression model for length of stay of ``Major cardiovascular interventions''. Boldface: p-values smaller than $0.05$ for the null hypothesis that the parameter equals zero. Abbreviations: c.d.$=$ complete data set, o.r.$=$ outliers removed, /G$=$ Gaussian errors.}
\begin{center}
\scalebox{0.75}{
\begin{tabular}[c]{lcccccc}
\hline
& $\mu$ & $10\beta_{1}$ & $10\beta_{2}$ & $10\gamma$ & $\sigma$ & $\lambda$ \\
\hline
1TML & \textbf{2.16} (0.11) & $-$1.84 (8.73) & \textbf{0.08} (0.02) & 0.09 (0.12) & 0.45 (0.08) & \textbf{$-$1.82} (0.14)\\
2TML & \textbf{2.16} (0.11) & $-$1.84 (8.71) & \textbf{0.08} (0.02) & 0.09 (0.12) & 0.45 (0.08) & \textbf{$-$1.82} (0.14)\\
1SWL & \textbf{2.16} (0.16) & $-$1.84 (8.10) & \textbf{0.08} (0.03) & 0.09 (0.12) & 0.45 (0.07) & \textbf{$-$1.82} (0.43)\\
2SWL & \textbf{2.16} (0.16) & $-$1.84 (8.12) & \textbf{0.08} (0.03) & 0.09 (0.12) & 0.45 (0.07) & \textbf{$-$1.82} (0.43)\\
ML c.d. &  2.05 (1.10) & 22.03 (123.00) & 0.07 (0.13) & $-$0.27 (2.04) & 0.45 (1.59) & $-$3.34 (17.54)\\
ML o.r. & \textbf{2.09} (0.10) & 1.11 (10.02) & \textbf{0.07} (0.02) & 0.04 (0.14) & 0.48 (0.10) & \textbf{$-$2.79} (0.58)\\
ML/G & 2.93 & 23.42 & 0.11 & $-$0.30 & -- & --\\
WML/G & 2.44 & 10.57 & 0.11 & $-$0.10 & -- & --\\
\hline
\end{tabular}
} %%% resize
\end{center}
\label{TabExaMCI}
\end{table}

\begin{figure}[ptb]
\begin{center}
\includegraphics[width=0.9\textwidth]{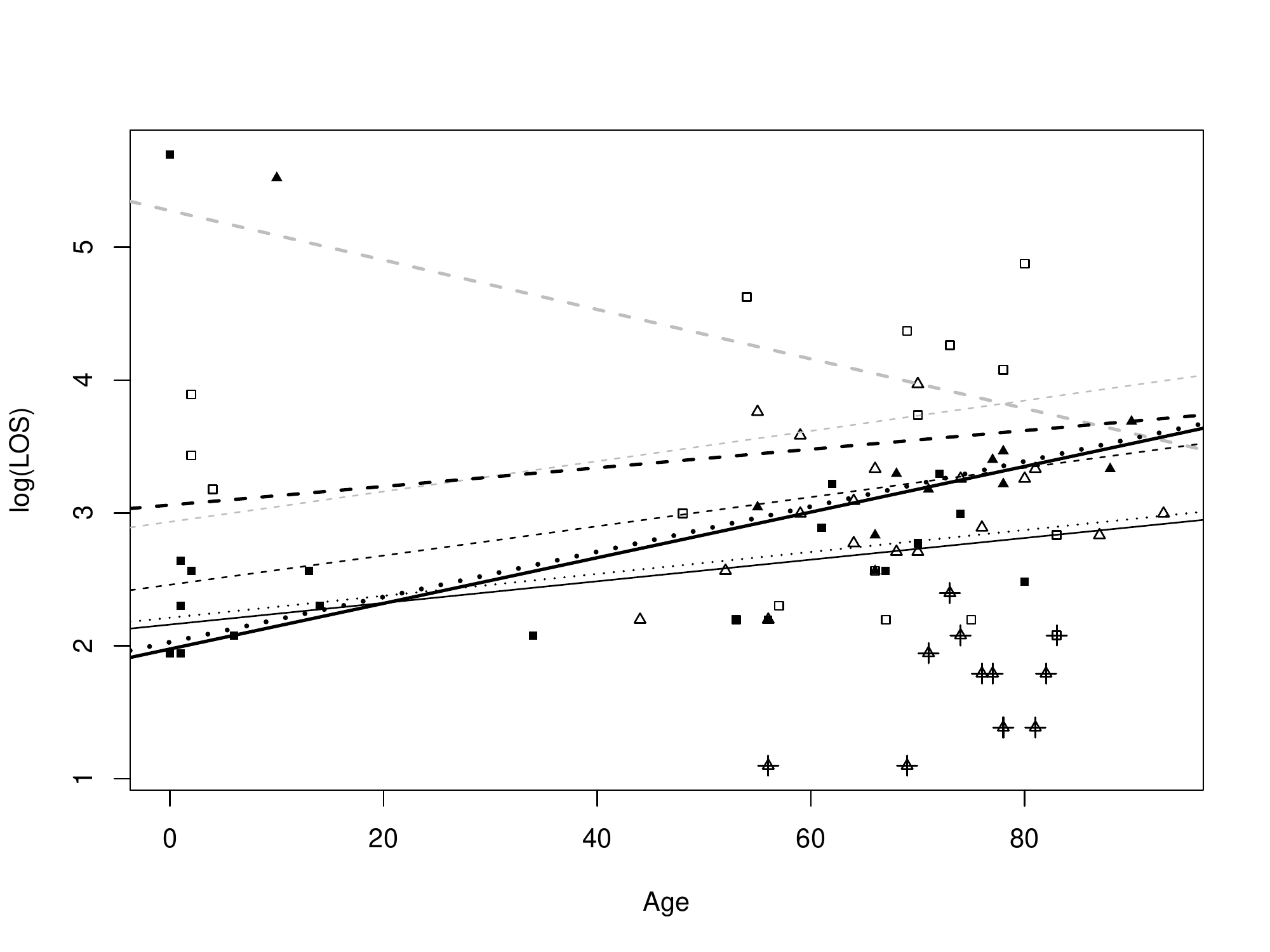}
\end{center}
\caption{Estimates of the regression models for length of stay of ``Major cardiovascular interventions''. Thin lines refer to planned admissions, while thick lines refer to emergency admission. ML is grey and dashed, WML/G is black and dashed, 2TML is black and solid, 2SWL is black and dotted. Filled marks represent complete observation; empty marks represent censored observation; square marks are planned admission; triangle are emergency admission. Outliers according to 2TML are marked with crosses.}
\label{FigExaMCI}
\end{figure}

\begin{figure}[ptb]
\begin{center}
\includegraphics[width=0.9\textwidth]{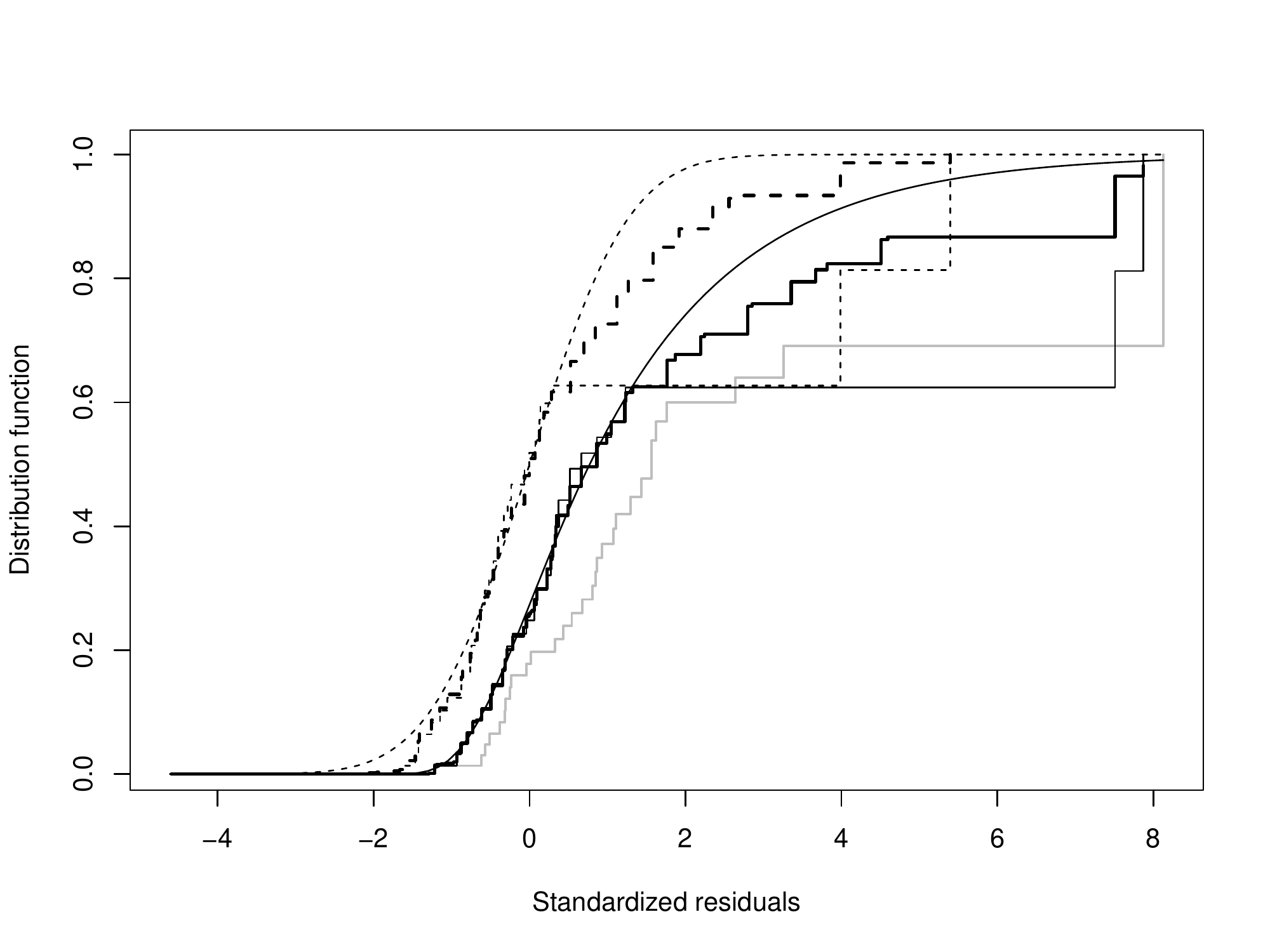}
\end{center}
\caption{Standardized residuals cdf of different estimates for the ``Major cardiovascular interventions'' data. Kaplan Meyer cdf of the ML residuals is grey; WML/G residual cdf is black and dashed; 2TML residual cdf is black and solid. Thin lines are Kaplan-Meyer cdf; thick lines are semiparametric cdf; the smooth line is the parametric cdf.}
\label{FigExaMCIFhat}
\end{figure}

\subsection{Minor bladder interventions}
In a second example, we consider a sample of $48$ stays for DRG ``Minor
bladder interventions''. The data are shown in Figure \ref{FigExaMBI}. Six patients were transferred to a different hospital before dismissal. Four young patients have surprisingly large values of LOS. We study the relationship between LOS and Age considering the model $y=\mu+\beta x+\sigma u$, where $x=$Age and $y=\log($LOS$)$. The ML, the 1TML, 2TML and 1SWL estimates and estimated standard errors are reported in Table \ref{TabExaMBI}; the corresponding prediction lines are drawn in Figure \ref{FigExaMBI}.

\begin{table}[ptb]
\caption{Estimates of the regression model for length of stay of ``Minor bladder interventions''. Abbreviations: c.d.$=$ complete data set; o.r.$=$ outliers removed.}
\begin{center}
\scalebox{0.75}{ 
\begin{tabular}[c]{lcccccc}
\hline
Method & $\mu$ & $10\beta$ & $\sigma$ & $\lambda$ &  & \\
\hline
1TML & 0.992 (0.365) &  0.127 (0.056) &
 0.480 (0.050) &  0.489 (0.170) &  & \\
2TML & 0.992 (0.365) &  0.127 (0.056) &
 0.482 (0.047) &  0.489 (0.172) &  & \\
1SWL & 0.992 (0.600) &  0.130 (0.080) &
 0.481 (0.063) &  0.490 (0.120) &  & \\
2SWL &  0.992 (0.606) &  0.131 (0.081) &
 0.484 (0.064) &  0.489 (0.120) &  & \\
ML c.d. & 2.633 (0.491) &  $-$0.113 (0.056) &
 0.686 (0.066) &  0.049 (0.565) &  & \\
ML o.r. & 0.525 (0.363) &  0.161 (0.052) &
 0.564 (0.069) & $-$0.449 (0.339) &  & \\
\hline
\end{tabular}
} %%% resize
\end{center}
\label{TabExaMBI}
\end{table}

According to ML, log(LOS) does not seem to depend on Age (the p-value is $0.0444$) and a Gauss distribution seems to be adequate (the p-value for $\lambda$ is $0.931$). The robust estimates are similar and provide a much larger slope (p-value $=0.019$). Moreover -- as it is expected
from this data -- they suggest a positive linear relationship and an asymmetric error model (the p-value for $\lambda$ is $0.0024$). Clearly, the outliers (those with weights equals to zero in 2TML are marked with crosses in Figure \ref{FigExaMBI}) have an important leverage effect on the ML coefficients and shape parameter. Removing the outliers, ML becomes similar to the robust estimates. In practice, this simple analysis suggests that the possibility of splitting this particular DRG into two groups should be further investigated.

\begin{figure}[ptb]
\begin{center}
\includegraphics[width=0.9\textwidth]{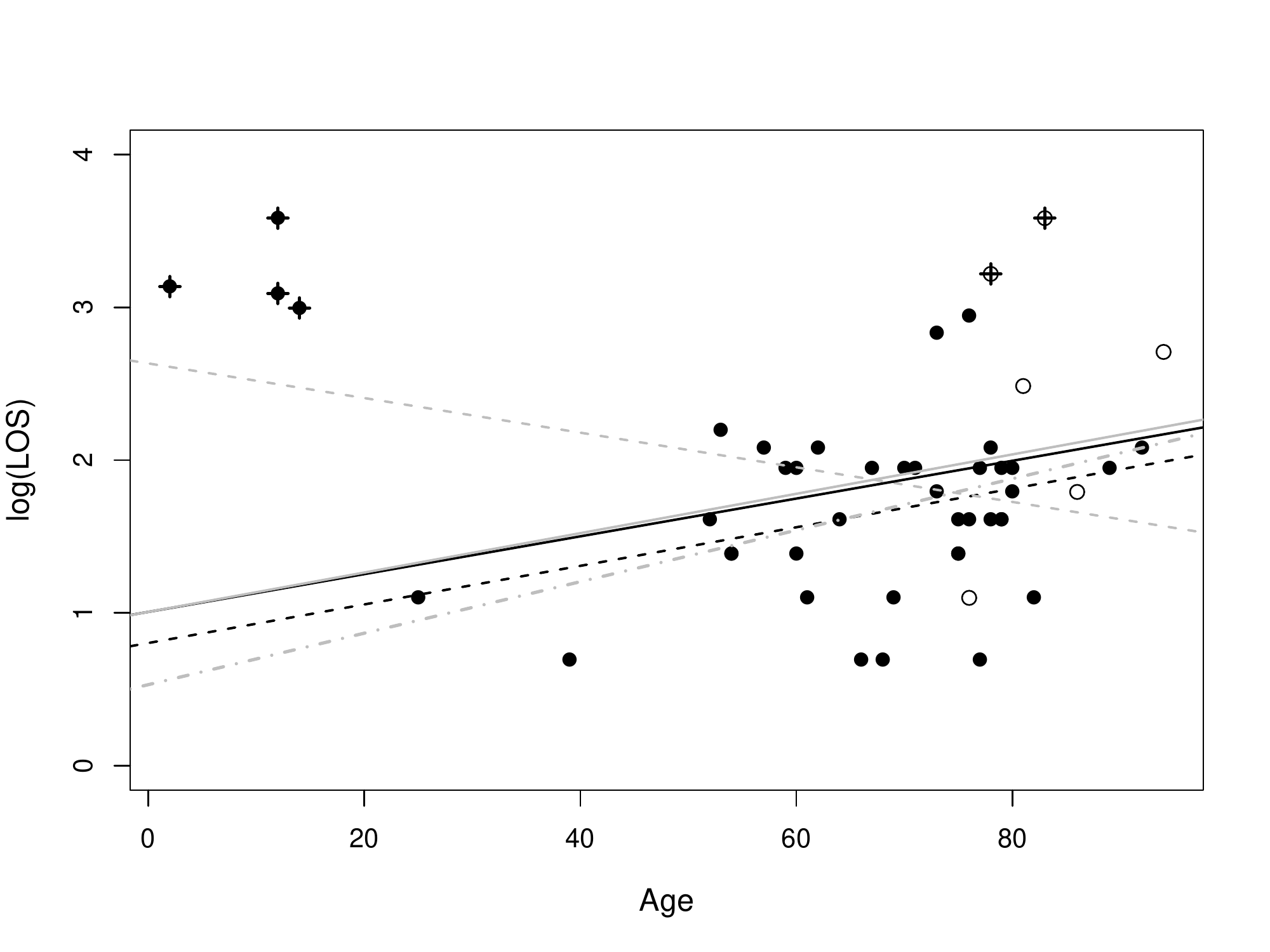}
\end{center}
\caption{Estimates of the regression models for length of stay of ``Major bladder interventions''. Maximum Likelihood based on the complete data set (ML c.d.) is grey and dashed, initial non parametric MM is black and dashed, Maximum Likelihood after removing outliers according to 2TML (ML o.r.) is gray and dash and dotted, 2TML is black and solid and 2SWL is gray and solid. Outliers according to 2TML are marked with a cross.}
\label{FigExaMBI}
\end{figure}

\section{Concluding remarks}
\label{SecConclusions} 
As mentioned in the introduction, the GLG model is a very flexible family of distributions which is used to describe asymmetrically distributed data in many real applications. In this paper, we considered estimators which are simultaneously robust and efficient for AFT models, when the errors follow a GLG distribution and the data may contain censored observations. The estimation procedures have two main components: an initial highly robust but not necessarily efficient estimator and a final efficient estimator which starts with the initial one.

We first considered the case, where no covariables are present and, in this case, we proposed an initial estimator that minimizes a $\tau$ scale of the differences between theoretical an empirical quantiles of order smaller than $(1-\alpha)$, where $0<\alpha<1$ is a trimming fraction. The final estimator is a one step weighted likelihood estimator, where the weights penalizing the outliers are derived from the initial estimator.

For the case, where covariables are present, the proposed estimators were derived in three steps. In a first step we used a regression MM-estimator as proposed in \citet{salibianbarrera2008} to obtain initial slope estimates and to compute the corresponding residuals. In a second step, we computed an initial estimator of the GLG parameters by applying the procedure for the no covariables case to these residuals. In the third step we obtained a final estimator of all the parameters using a one step truncated ML starting with the initial estimator.

We provided asymptotic results and extensive Monte Carlo results showing that the final estimators are highly efficient and maintain the same robustness level as the initial ones.

%% \section{Supplemental Material}
%% \label{SecSupplemental}

%% Supplemental Material is available. Section \ref{SecSMConsistencyTQtau} provides proofs for consistency and $n^{1/2}$-consistency of the TQ$\tau$ estimator; Section \ref{SecSMAsymptotic1TML} provide a proof of the asymptotic normality of the 1TML estimator; Section \ref{SecSM1SWL} extends the one step Weighted Likelihood estimator of \citet{agostinelli2014a} to the regression case with censored observations.

\section*{Acknowledgements}
All statistical analysis were performed on SCSCF (www.dais.unive.it/scscf), a multiprocessor cluster system owned by Ca' Foscari University of Venice running under GNU/Linux.

This research was partially supported by the Italian-Argentinian project ``Metodi robusti per la previsione del costo e della durata della degenza ospedaliera'' funded by the collaboration program MINCYT-MAE AR14MO6.

V\'{\i}ctor Yohai research was also partially supported by Grants
20020130100279 from Universidad of Buenos Aires, PIP 112-2008-01-00216 and 112-2011-01-00339 from CONICET and PICT 2011-0397 from ANPCYT.

\clearpage

\appendix
\section{Appendix}
\label{appendix}

The Appendix contains in subsection \ref{SecSMConsistencyTQtau} the proof for the  consistency of the TQ$\tau$ estimator for both cases: without and with covariates, in subsection \ref{SecSMAsymptotic1TML} the derivation of the asymptotic distribution of the 1TML and 2TML estimators and  in subection \ref{SMSec1SWL} the definition of the one step weighted likelihood estimators.

\subsection{Consistency of the TQ$\tau$ estimator}
\label{SecSMConsistencyTQtau} 
Consistency and $n^{1/2}$-consistency are proved for the TQ$\tau$ estimator for the case without covariates and for the regression case in the next two subsubsections.

\subsubsection{The case without covariables}
Consider the TQ$\tau$ estimator with trimming proportion $\alpha$ defined in Section \ref{SecTQtau} to estimate the parameter $\mathbf{\theta}$ of a GLG distribution. We need the following assumptions.

\begin{enumerate}[label=\textbf{B\arabic*}]
\item \label{assSMB1} $\Theta$ is a compact set.
\item \label{assSMB2} For all $\mathbf{\theta}$ and all $y\in\mathbb{R}$ we have $0<F_{\mathbf{\theta}}(y)<1$.
\item \label{assSMB3} $F_{\mathbf{\theta}}(y)$ has a continuous and bounded density $f_{\mathbf{\theta}}(y)>0$ in $y$ and $\mathbf{\theta}$.
\item \label{assSMB4} Given $\mathbf{\theta}_{1}\neq\mathbf{\theta}_{2}$ there is only a finite number of values $y$ such that $F_{\mathbf{\theta}_{1}}(y)=F_{\mathbf{\theta}_{2}}(y)$.
\item \label{assSMB6} Let $Q(u,\mathbf{\theta})$ be defined as the unique value $q$ such that $F_{\mathbf{\theta}}(q)=u$.
\end{enumerate}

Note that \textbf{\ref{assSMB1}}-\textbf{\ref{assSMB3}} imply that $Q(u,\mathbf{\theta})$ is continuous in $u$ and $\mathbf{\theta}$ and strictly increasing in $u$. Put $\Delta Q(u,\mathbf{\theta})=\partial Q(u,\mathbf{\theta})/\partial\mathbf{\theta}$ and $E_{0}=\{\mathbf{e}\in\mathbb{R}^{p}:||\mathbf{e}||=1\}$. Then, for all $\mathbf{e}\in E_{0}$, there is only a finite number of values $u$ such that $\Delta Q(u,\mathbf{\theta}_{0})^{\top}\mathbf{e}=0$.

\begin{theorem}
\label{TeoSMConsistencyTQtau} Assume $\rho_{1}$ and $\rho_{2}$ satisfy assumption \ref{assA1} and \ref{assSMB1}-\ref{assSMB6}. Let $\mathbf{\tilde{\theta}}$ be the TQ$\tau$ estimator. Then,
\begin{equation*}
\tilde{\mathbf{\theta}} \rightarrow\mathbf{\theta} \qquad\text{a.s.}, n
\rightarrow\infty
\end{equation*}
and
\begin{equation*}
n^{1/2}\left(  \tilde{\mathbf{\theta}} - \mathbf{\theta} \right)  = O_{P}(1)
.
\end{equation*}

\end{theorem}

\begin{proof}[Proof]
The proof of consistency is similar to the proof of Theorem 1 in \citet{agostinelli2014b} while the proof of $n^{1/2}$ consistency is similar to the proof of Theorem 2 in \citet{agostinelli2014b}. The only difference is that we replace the empirical distribution $F_{n}$ by the Kaplan-Meier distribution $\tilde{F}_{n}$. Note that the only property of $F_{n}$ that was used in the proof of Theorem 2 is 
\begin{equation*}
n^{1/2}\sup_{t}\left\vert F_{n}^{-1}(t)-F_{0}^{-1}(t)\right\vert =O_{p}(1).
\end{equation*}
Then, in order to prove this Theorem, it is enough to prove that
\begin{equation}
n^{1/2}\sup_{t}\left\vert \tilde{F}_{n}^{-1}(t)-F_{0}^{-1}(t)\right\vert=O_{p}(1), \label{equSMsqFNtilde}
\end{equation}
and this is proved in Lemma \ref{LemSMsqFNtilde} below. Since we are considering an $\alpha$-trimmed Q$\tau$, we also need a slightly modified
version of Lemma 2 in \citet{agostinelli2014b}, where instead of considering the range of quantiles in the interval $[0.1,0.9]$, we consider the range $[0.10,1-\alpha]$. The proof is exactly the same and it is omitted.
\end{proof}

Suppose now that $F_{n}$ is a sequence of estimators of a distribution
function $F_{0}$ with support $(a,b)$, where $a$ may be $-\infty$ and $b$ may be $+\infty$. We consider the random variables $V_{n}=n^{1/2}\sup_{x}|F_{n}(x)-F_{0}(x)|$ and $S_{n}=n^{1/2}\sup_{\alpha\leq u\leq\beta}|F_{n}^{-1}(u)-F_{0}^{-1}(u)|$, where $0<\alpha<\beta<1$. The following assumptions are required.

\begin{enumerate}[label=\textbf{C\arabic*}]
\item \label{assSMD1} $V_{n}$ is bounded in probability.
\item \label{assSMD2} $f_{0}(x)=F_{0}^{\prime}(x)$ is continuous and positive in $(a,b)$.
\end{enumerate}

\begin{lemma}
\label{LemSM2} Assume \ref{assSMD1} and \ref{assSMD2}. Then, $S_{n}$ is bounded in probability.
\end{lemma}

\begin{proof}[Proof]
Take $\delta=\min(F_{0}^{-1}(\alpha)/2,(1-F_{0}^{-1}(\beta))/2)$ and let
\begin{equation}
\eta=\inf_{F^{-1}(\alpha)-\delta\leq x\leq F^{-1}(\beta)+\delta}f_{0}(x) . \label{equSMeta}
\end{equation}
By \ref{assSMD1}, for every $\varepsilon>0$ there exists $K_{1}$ such
\begin{equation}
P(V_{n}>K_{1})<\varepsilon. \label{equSMVnK1}
\end{equation}
Take $n_{0}$ such that
\begin{equation}
K_{1}/n_{0}^{1/2}<\delta\label{equSMn0}
\end{equation}
and $K_{0}$, such that
\begin{equation}
K_{0} > K_{1}/\eta. \label{equSMK0K1}
\end{equation}
We will show that, for $n \geq n_{0}$,
\begin{equation*}
P(S_{n}>K_{0})<\varepsilon.
\end{equation*}
To prove this, it is enough to show that for $n \geq n_{0}$
\begin{equation}
\{S_{n}>K_{0}\} \subset\{V_{n}>K_{1}\} . \label{equSMunion3}
\end{equation}
By (\ref{equSMVnK1}), to prove this it is enough to show that, for $n \geq n_{0}$ and for all $u$ such that $\alpha\leq u \leq\beta$, we have
\begin{equation}
V_{n}\leq K_{1}\Longrightarrow n^{1/2}(F_{n}^{-1}(u)-F_{0}^{-1}(u))\leq K_{0} \label{equSMimp11}
\end{equation}
and
\begin{equation}
V_{n}\leq K_{1}\Longrightarrow n^{1/2}(F_{n}^{-1}(u)-F_{0}^{-1}(u))\geq-K_{0}. \label{equSMimp22}
\end{equation}
Suppose that $n\geq n_{0}$, $V_{n}\leq K_{1}$ and $\alpha\leq u \leq\beta$. Then, using (\ref{equSMeta}), (\ref{equSMn0}), (\ref{equSMK0K1}), we get
\begin{align*}
& F_{n}\left(  F_{0}^{-1}(u)+\frac{K_{0}}{n^{1/2}}\right)  >F_{0}\left( F_{0}^{-1}(u)+\frac{K_{0}}{n^{1/2}}\right)  -\frac{K_{1}}{n^{1/2}} \\
&  >F_{0}(F_{0}^{-1}(u))+\frac{\eta K_{0}}{n^{1/2}}-\frac{K_{1}}{n^{1/2}} > u+\frac{\eta K_{0}-K_{1}}{n^{1/2}}>u,
\end{align*}
and therefore $F_{n}^{-1}(u)\leq F_{0}^{-1}(u)+K_{0}/n^{1/2}$. Then (\ref{equSMimp11}) holds. The proof of (\ref{equSMimp22}) is similar.
\end{proof}

\begin{lemma}
\label{LemSMsqFNtilde} Assume $\rho_{1}$ and $\rho_{2}$ satisfy assumption \ref{assA1} and \ref{assSMB1}-\ref{assSMB6} and \ref{assSMD1}-\ref{assSMD2} then
\begin{equation*}
n^{1/2}\sup_{t}\left\vert \tilde{F}_{n}^{-1}(t)-F_{0}^{-1}(t)\right\vert = O_{p}(1).
\end{equation*}
\end{lemma}

\begin{proof}[Proof]
Let $\tilde{F}_{n}$ be the Kaplan-Meier distribution applied to $Y_{n}^{\ast}=(y_{1}^{\ast},\cdots,y_{n}^{\ast})$ and $\Delta_{n}$. \citet{breslow1974} showed that
\begin{equation*}
n^{1/2}\sup_{t}\left\vert \tilde{F}_{n}(t)-F_{0}(t)\right\vert =O_{p}(1),
\end{equation*}
and hence, by Lemma \ref{LemSM2}, we have the result.
\end{proof}

\subsubsection{The case with covariables}
We now move to the regression estimators defined in Section \ref{SecRegression}. We suppose that $(y_{i}^{\ast},\mathbf{x}_{i},\delta_{i})$ ($1\leq i\leq n$, $\mathbf{x}_{i}\in\mathbb{R}^{p}$) is a sample of observations which follow the model (\ref{model}), i.e.,
\begin{equation*}
y_{i}=\mu_{0}+\mathbf{\beta}_{0}^{\top}\mathbf{x}_{i}+\sigma_{0} u_{i}, \qquad i=1,\ldots,n,
\end{equation*}
and $u_{i} \sim f_{0,1,\lambda_{0}}$. Let $\mathbf{\bar{\beta}}_{n}$ be an estimator of $\mathbf{\beta}_{0}$. We consider the following assumptions
\begin{enumerate}[label=\textbf{D\arabic*}]
\item \label{assSME1} $n^{1/2}(\mathbf{\bar{\beta}}_{n}-\mathbf{\beta}_{0})$ is bounded in probability.
\item \label{assSME2} There exist $C$ such that for all $i$ we have $||\mathbf{x}_{i}||\leq C$.
\end{enumerate}
We define
\begin{align*}
\upsilon_{i}  &  = y_{i}-\mathbf{\beta}_{0}^{\top}\mathbf{x}_{i}=\mu_{0}+\sigma_{0}u_{i},\\
\upsilon_{i}^{\ast}  &  = y_{i}^{\ast}-\mathbf{\beta}_{0}^{\top}\mathbf{x}_{i}=\min(\upsilon_{i},c_{i}-\mathbf{\beta}_{0}^{\top}\mathbf{x}_{i}),
\end{align*}
and the residuals
\begin{equation*}
\omega_{i}^{\ast} = y_{i}^{\ast}-\mathbf{\bar{\beta}}_{n}^{\top}\mathbf{x}_{i}.
\end{equation*}
Then,
\begin{equation*}
\omega_{i}^{\ast}=\min(\upsilon_{i},c_{i}-\mathbf{\beta}_{0}^{\top} \mathbf{x}_{i})+(\mathbf{\beta}_{0}-\mathbf{\bar{\beta}}_{n})^{\top} \mathbf{x}_{i}=\upsilon_{i}^{\ast}+(\mathbf{\bar{\beta}}_{n}-\mathbf{\beta}_{0})^{\top}\mathbf{x}_{i}.
\end{equation*}
We consider that the estimate $\tilde{\mathbf{\theta}}$ of $\mathbf{\theta}_{0} = (\mu_{0}, \sigma_{0}, \lambda_{0})$ is the TQ$\tau$ based on the residuals $(\omega_{i}^{\ast}, \delta_{i})$, $1 \le i \le n$ and we let $\tilde{\mathbf{\gamma}} = (\bar{\mathbf{\beta}}, \tilde{\mathbf{\theta}})$ and $\mathbf{\gamma}_{0} = (\mathbf{\beta}_{0}, \mathbf{\theta}_{0})$. Suppose that, using a sample $X_{n}=(x_{1},\ldots,x_{n})$ and the censorship indicators $\Delta_{n}=(\delta_{1},\ldots,\delta_{n})$, we have an estimator $F_{n}^{X_{n},\Delta_{n}}(x)$ of a distribution $F_{0}$. We assume that
\begin{enumerate}[label=\textbf{E\arabic*}]
\item \label{assSMC1} $n^{1/2}\sup_{x}|F_{n}^{X_{n},\Delta_{n}}(x)-F_{0}(x)|$ is bounded in probability.
\item \label{assSMC2} Suppose that given two samples $X_{n}=(x_{1},\ldots,x_{n})$ and $X_{n}^{\ast}=(x_{1}^{\ast},\ldots,x_{n}^{\ast})$ such that $x_{i}\geq x_{i}^{\ast}$ for all $i$, then $F_{n}^{X_{n},\Delta_{n}}(x)\leq$ $F_{n}^{X_{n}^{\ast},\Delta_{n}}(x)$.
\item \label{assSMC3} Given a real number $z$ and a sample $X_{n}=(x_{1},\ldots,x_{n})$ , let $X_{n}+z$ $=(x_{1}+z,\ldots,x_{n}+z)$. Then $F_{n}^{X_{n}+z,\Delta_{n}}(x)=F_{n}^{X_{n},\Delta_{n}}(x-z)$.
\item \label{assSMC4} Let $Q_{n}=(q_{1},\ldots,q_{n})$ be a sequence of random variables such that \hbox{   } $n^{1/2}\sup_{1\leq i\leq n}|q_{i}|$ is bounded in probability.
\end{enumerate}

\begin{theorem}
\label{TeoSMConsistencyTQtaureg} Assume \ref{assSME1}-\ref{assSME2} and \ref{assSMC1}-\ref{assSMC4} then $\tilde{\mathbf{\gamma}}$ is a consistent estimator of $\mathbf{\gamma}_{0}$ and furthermore
\begin{equation*}
n^{1/2}\sup_{t}\left\vert \tilde{\mathbf{\gamma}} -\mathbf{\gamma}_{0} \right\vert =O_{p}(1) .
\end{equation*}
\end{theorem}
In order to prove Theorem \ref{TeoSMConsistencyTQtaureg} we need the following Lemma \ref{LemSM1}.
\begin{lemma}
\label{LemSM1} Assume \ref{assSMC1}-\ref{assSMC4}. Then, $U_{n}=n^{1/2} \sup_{x}|F_{n}^{X_{n}+Q_{n},\Delta_{n}}(x)-F_{0}(x)|$ is bounded in probability.
\end{lemma}

\begin{proof}[Proof]
We have to prove that given $\varepsilon>0$ there exists $K_{0}$ such that
\begin{equation}
P(\sup_{n}U_{n}>K_{0})<\varepsilon. \label{equSMtoprove}
\end{equation}
Let $V_{n}=n^{1/2}\sup_{x}|F_{n}^{X_{n},\Delta_{n}}(x)-F_{0}(x)|$. By \ref{assSMC1}, there exists $K_{1}$ such that
\begin{equation}
P(\sup_{n}V_{n}>K_{1})<\varepsilon/3 . \label{equSMass1}
\end{equation}
Let $W_{n}=n^{1/2}\sup_{1\leq i\leq n}|q_{i}|$. By \ref{assSMC4}, there exists $K_{2}$ such that
\begin{equation}
P(\sup_{n}W_{n}>K_{2})<\varepsilon/3 . \label{equSMass5}
\end{equation}
Let $K_{3}=\sup_{x}f_{0}(x)$ and put $K_{0}=K_{1}+K_{2}K_{3}$. Then, by \ref{assSMC2} and \ref{assSMC3} we have
\begin{align*}
&  P(U_{n}>K_{0})=P\left(  n^{1/2}\sup_{x}\left\vert F_{n}^{X_{n}+Q_{n},\Delta_{n}}(x)-F_{0}(x)\right\vert >K_{0}\right) \\
&  \leq P\left(  n^{1/2}\sup_{x}\left\vert F_{n}^{X_{n}+\frac{K_{2}}{n^{1/2}},\Delta_{n}}(x)-F_{0}(x)\right\vert >K_{0}\right) \\
&  +P\left(  n^{1/2}\sup_{x}\left\vert F_{n}^{X_{n}-\frac{K_{2}}{n^{1/2}},\Delta_{n}}(x)-F_{0}(x)\right\vert >K_{0}\right)  +P(W_{n}>K_{2}).
\end{align*}
By \ref{assSMC3} we have
\begin{equation}
n^{1/2}\sup_{x}\left\vert F_{0}\left(  x-\frac{K_{2}}{n^{1/2}}\right)-F_{0}(x)\right\vert \leq K_{2}K_{3}. \label{equSMLasK}
\end{equation}
Then, using (\ref{equSMLasK}), we get
\begin{align}
&  P\left(  n^{1/2}\sup_{x}\left\vert F_{n}^{X_{n}+\frac{K_{2}}{n^{1/2}},\Delta_{n}}(x)-F_{0}(x)\right\vert >K_{0}\right) \nonumber\\
&  =P\left(  n^{1/2}\sup_{x}\left\vert F_{n}^{X_{n},\Delta_{n}}\left(x-\frac{K_{2}}{n^{1/2}}\right)  -F_{0}(x)\right\vert >K_{1}+K_{2}K_{3}\right) \nonumber\\
&  \leq P\left(  n^{1/2}\sup_{x}\left\vert F_{n}^{X_{n},\Delta_{n}}\left(x-\frac{K_{2}}{n^{1/2}}\right)  -F_{0}\left(  x-\frac{K_{2}}{n^{1/2}}\right) \right\vert \right. \nonumber\\
&  +\left.  n^{1/2}\sup_{x}\left\vert F_{0}\left(  x-\frac{K_{2}}{n^{1/2}}\right)  -F_{0}(x)\right\vert >K_{1}+K_{2}K_{3}\right) \nonumber\\
&  \leq P\left(  n^{1/2}\sup_{x}\left\vert F_{n}^{X_{n},\Delta_{n}}\left(x-\frac{K_{2}}{n^{1/2}}\right)  -F_{0}\left(  x-\frac{K_{2}}{n^{1/2}}\right) \right\vert >K_{1}\right) \nonumber\\
&  \leq P\left(  V_{n}>K_{1}\right)  \leq\frac{\varepsilon}{3}. \label{equSMfin1}
\end{align}
Similarly we can prove that
\begin{equation}
P\left(  n^{1/2}\sup_{x}\left\vert F_{n}^{X_{n}-\frac{K_{2}}{n^{1/2}},\Delta_{n}}(x)-F_{0}(x)\right\vert >K_{0}\right)  \leq\frac{\varepsilon}{3}. \label{equSMfin2}
\end{equation}
Then, from (\ref{equSMass5}), (\ref{equSMfin1}) and (\ref{equSMfin2}) we get (\ref{equSMtoprove}).
\end{proof}

\begin{proof}[Proof of Theorem \ref{TeoSMConsistencyTQtaureg}]
Let $F_{n}$ be the Kaplan-Meier distribution applied to $V_{n}^{\ast}=(\upsilon_{1}^{\ast},\ldots,\upsilon_{n}^{\ast})$ and $\Delta_{n}$, let $\tilde{F}_{n}$ be the Kaplan-Meier distribution applied to $(\Omega_{n}^{\ast},\Delta_{n})$, where $\Omega_{n}^{\ast}=(\omega_{1}^{\ast},\ldots,\omega_{n}^{\ast})$ and $\Delta_{n}=(\delta_{1},\ldots,\delta_{n})$ and let $F_{0}$ be the distribution of the $\upsilon_{i}$s. The proof is similar to that of Theorem \ref{TeoSMConsistencyTQtau} and hence it is sufficient to prove that
\begin{equation*}
n^{1/2}\sup_{t}|\tilde{F}_{n}^{-1}(t)-F_{0}^{-1}(t)|=O_{P}(1).
\end{equation*}
\citet{breslow1974} showed that
\begin{equation*}
n^{1/2}\sup_{t}\left\vert F_{n}(t)-F_{0}(t)\right\vert =O_{p}(1).
\end{equation*}
Then, since the Kaplan-Meier distribution satisfies \ref{assSMC2} and \ref{assSMC3}, Lemma \ref{LemSM1} implies
\begin{equation*}
n^{1/2}\sup_{t}|\tilde{F}_{n}(t)-F_{0}(t)|=O_{P}(1).
\end{equation*}
On the other hand $\Omega_{n}^{\ast}=V_{n}^{\ast}+Q_{n}$ where $Q_{n}=(\mathbf{\beta}_{0}-\mathbf{\bar{\beta}}_{n})^{\text{T}}\mathbf{x}_{i}$ satisfies \ref{assSMC4} and by Lemma \ref{LemSM2} the result holds.
\end{proof}

\subsection{Asymptotic distribution of the 1TML estimator}
\label{SecSMAsymptotic1TML} 
In this Section we study the asymptotic behavior of the 1TML estimator. Note that the asymptotic behavior of the 2TML is obtained in a very similar way and hence it is not reported. We suppose that $\tilde{\mathbf{\gamma}}$ is an is an initial $n^{1/2}$ consistent estimator, e.g. the TQ$\tau$ estimator, and and that $\mathbf{\gamma}_{0}$ is the true parameter value. We need the following assumptions:
\begin{enumerate}[label=\textbf{F\arabic*}]
\item \label{assSMF1} $n^{1/2}(\tilde{\mathbf{\gamma}} - \mathbf{\gamma}_{0})$ is bounded in probability.
\item \label{assSMF2} The matrix $\mathbf{G}(\mathbf{\gamma}) = E\left(\nabla_{\mathbf{\gamma}} \mathbf{v}(\mathbf{x}, y,\delta,\mathbf{\gamma})\right)  $ (where $\mathbf{v}(\mathbf{x}, y,\delta,\mathbf{\gamma})$ is defined in equation (\ref{equML3})) is non singular.
\item \label{assSMF3} The vector $\mathbf{x}$ has second order moments.
\end{enumerate}

\begin{theorem}
\label{TeoSMAsymptotic1TMLreg} Assume (i) the weight function $w$ satisfies \ref{assA2}, (ii) the function $w$ is continuously differentiable (iii) \ref{assSMF1}-\ref{assSMF3} holds, then the 1TML estimator $\hat{\mathbf{\gamma}}_{1}$ is such that
\begin{equation*}
\sqrt{n}(\hat{\mathbf{\gamma}}_{1} - \mathbf{\gamma}_{0}) \overset{\mathcal{L}}{\rightarrow} N(0, \mathbf{G}(\mathbf{\gamma}_{0})^{-1} \mathbf{M}(\mathbf{\gamma}_{0})\mathbf{G}(\mathbf{\gamma}_{0})^{-\top}).
\end{equation*}
\end{theorem}
To prove the Theorem we need some additional notations and the lemma stated below. Let
\begin{equation*}
\mathbf{\eta}(\mathbf{z}_{i},\mathbf{\gamma},\tilde{\mathbf{\gamma}},\vartheta^{\ast})=(\eta_{1}(\mathbf{z}_{i},\mathbf{\gamma},\tilde{\mathbf{\gamma}},\vartheta^{\ast}),\cdots,\eta_{p+3}(\mathbf{z}_{i},\mathbf{\gamma},\tilde{\mathbf{\gamma}},\vartheta^{\ast}))^{\top}
\end{equation*}
so that, for given $\mathbf{\gamma}_{1}$ and $\mathbf{\gamma}_{2}$,
\begin{equation*}
\eta_{k}(\mathbf{z}_{i},\mathbf{\gamma}_{1},\mathbf{\gamma}_{2},\vartheta^{\ast})=E_{\mathbf{\gamma}_{1}}\left[  w\left(  \mathbf{z},\mathbf{\gamma}_{2},\vartheta^{\ast}\right)  v_{k}\left(  \mathbf{z},\mathbf{\gamma}_{1}\right)  |\mathbf{z}_{i}\right]  \qquad k=1,\cdots,(p+3),
\end{equation*}
where $v_{1},\cdots,v_{p+3}$ are the score functions evaluated at $\mathbf{\gamma}_{1}$. Then, the functions $u_{k}(\mathbf{\gamma},\tilde{\mathbf{\gamma}},\vartheta^{\ast})$ in (\ref{equTMLreg}) can be written as follows:
\begin{equation*}
u_{k}(\mathbf{\gamma},\tilde{\mathbf{\gamma}},\vartheta^{\ast})=\frac{1}{n}\sum_{i=1}^{n}\eta_{k}(\mathbf{z}_{i},\mathbf{\gamma},\tilde{\mathbf{\gamma}},\vartheta^{\ast}),\qquad k=1,\cdots,p+3.
\end{equation*}
In addition, for given $\mathbf{\gamma}_{1}$, $\mathbf{\gamma}_{2}$, $\mathbf{\gamma}_{3}$, let
\begin{equation*}
\mathbf{h}(\mathbf{\gamma}_{1},\mathbf{\gamma}_{2},\mathbf{\gamma}_{3},\vartheta^{\ast})=(h_{1}(\mathbf{\gamma}_{1},\mathbf{\gamma}_{2},\mathbf{\gamma}_{3},\vartheta^{\ast}),\cdots,h_{p+3}(\mathbf{\gamma}_{1},\mathbf{\gamma}_{2},\mathbf{\gamma}_{3},\vartheta^{\ast}))^{\top},
\end{equation*}
where
\begin{equation*}
h_{k}(\mathbf{\gamma}_{1},\mathbf{\gamma}_{2},\mathbf{\gamma}_{3},\vartheta^{\ast})=E_{\mathbf{\gamma}_{1}}\left[  \eta_{k}(\mathbf{z},\mathbf{\gamma}_{2},\mathbf{\gamma}_{3},\vartheta^{\ast})\right] , \qquad k=1,\cdots,p+3.
\end{equation*}
We also have:
\begin{equation*}
h_{k}(\mathbf{\gamma},\tilde{\mathbf{\gamma}},\mathbf{\gamma},\vartheta^{\ast})=E_{\mathbf{\gamma}}\left[  w\left(  \mathbf{z},\tilde{\mathbf{\gamma}},\vartheta^{\ast}\right)  v_{k}\left(  \mathbf{z},\mathbf{\gamma}\right) \right]  ,\qquad k=1,\cdots,p+3.
\end{equation*}
Therefore, equations (\ref{equTMLreg}) can be written as
\begin{equation*}
\frac{1}{n}\sum_{i=1}^{n}\mathbf{\eta}(\mathbf{z}_{i},\mathbf{\gamma},\tilde{\mathbf{\gamma}},\vartheta^{\ast})=\mathbf{h}(\tilde{\mathbf{\gamma}},\tilde{\mathbf{\gamma}},\tilde{\mathbf{\gamma}},\vartheta^{\ast}).
\end{equation*}
Finally, notice that
\begin{equation*}
\mathbf{J}(\tilde{\mathbf{\gamma}},\vartheta^{\ast})=\frac{1}{n}\sum_{i=1}^{n}\nabla_{\mathbf{\gamma}}\mathbf{\eta}(\mathbf{z}_{i},\mathbf{\gamma},\tilde{\mathbf{\gamma}},\vartheta^{\ast})|_{\mathbf{\gamma}=\tilde{\mathbf{\gamma}}}.
\end{equation*}

\begin{lemma}
\label{LemSMAsymptotic1TML1} Under the assumptions of Theorem \ref{TeoSMAsymptotic1TMLreg}, we have
\begin{equation*}
\sqrt{n} | \mathbf{h}(\tilde{\mathbf{\gamma}}, \tilde{\mathbf{\gamma}},\tilde{\mathbf{\gamma}}, \vartheta^{\ast}) - \mathbf{h}(\mathbf{\gamma}_{0},\tilde{\mathbf{\gamma}}, \mathbf{\gamma}_{0}, \vartheta^{\ast}) | \overset{P}{\rightarrow} \mathbf{0} \qquad n \rightarrow\infty.
\end{equation*}
\end{lemma}

\begin{proof}[Proof]
We prove the lemma by showing that
\begin{equation} \label{equSMTMLh0}
\sqrt{n} | \mathbf{h}(\tilde{\mathbf{\gamma}},\tilde{\mathbf{\gamma}}, \tilde{\mathbf{\gamma}}, \vartheta^{\ast}) - \mathbf{h}(\mathbf{\gamma}_{0}, \mathbf{\gamma}_{0}, \mathbf{\gamma}_{0}, \vartheta^{\ast}) | \overset{P}{\rightarrow} \mathbf{0} \qquad n \rightarrow\infty,
\end{equation}
and
\begin{equation} \label{equSMTMLh1}
\sqrt{n} | \mathbf{h}(\mathbf{\gamma}_{0}, \mathbf{\gamma}_{0}, \mathbf{\gamma}_{0}, \vartheta^{\ast}) - \mathbf{h}(\mathbf{\gamma}_{0}, \tilde{\mathbf{\gamma}}, \mathbf{\gamma}_{0}, \vartheta^{\ast}) | \overset{P}{\rightarrow} \mathbf{0} \qquad n \rightarrow\infty.
\end{equation}
Let us consider (\ref{equSMTMLh0}). We can write
\begin{equation*}
\sqrt{n} (\mathbf{h}(\tilde{\mathbf{\gamma}}, \tilde{\mathbf{\gamma}}, \tilde{\mathbf{\gamma}}, \vartheta^{\ast}) - \mathbf{h}(\mathbf{\gamma}_{0}, \mathbf{\gamma}_{0}, \mathbf{\gamma}_{0}, \vartheta^{\ast})) = \nabla_{\mathbf{\gamma}} \mathbf{h}(\mathbf{\gamma}, \mathbf{\gamma}, \mathbf{\gamma}, \vartheta^{\ast}) |_{\mathbf{\gamma}=\mathbf{\gamma}^{\ast}} \sqrt{n}(\tilde{\mathbf{\gamma}} - \mathbf{\gamma}),
\end{equation*}
where $\mathbf{\gamma}^{\ast}$ is between $\tilde{\mathbf{\gamma}}$ and $\mathbf{\gamma}_{0}$. According to \ref{assSMF1}, the second factor is bounded in probability. The first factor is asymptotically zero since
\begin{align*}
\lim_{n \rightarrow\infty} \nabla_{\mathbf{\gamma}} \mathbf{h}(\mathbf{\gamma}, \mathbf{\gamma}, \mathbf{\gamma}, \vartheta^{\ast}) |_{\mathbf{\gamma}=\mathbf{\gamma}^{\ast}}  &  = \nabla_{\mathbf{\gamma}} \lim_{n \rightarrow\infty} \mathbf{h}(\mathbf{\gamma}, \mathbf{\gamma}, \mathbf{\gamma}, \vartheta^{\ast}) |_{\mathbf{\gamma}=\mathbf{\gamma}^{\ast}}\\
&  = \nabla_{\mathbf{\gamma}} \mathbf{h}(\mathbf{\gamma}, \mathbf{\gamma}, \mathbf{\gamma}, 0) |_{\mathbf{\gamma}=\mathbf{\gamma}_{0}}\\
&  = \nabla_{\mathbf{\gamma}} E_{\mathbf{\gamma}}(\mathbf{v}(\mathbf{z},\mathbf{\gamma})) |_{\mathbf{\gamma}=\mathbf{\gamma}_{0}}
\end{align*}
but $E_{\mathbf{\gamma}}(\mathbf{v}(\mathbf{z}, \mathbf{\gamma})) = \mathbf{0}$ for all $\mathbf{\gamma}$. For (\ref{equSMTMLh1}) we can write
\begin{align*}
&  \sqrt{n} (\mathbf{h}(\mathbf{\gamma}_{0}, \mathbf{\gamma}_{0},\mathbf{\gamma}_{0}, \vartheta^{\ast}) - \mathbf{h}(\mathbf{\gamma}_{0},\tilde{\mathbf{\gamma}}, \mathbf{\gamma}_{0}, \vartheta^{\ast}) )\\
&  = \sqrt{n} E_{\mathbf{\gamma}_{0}} (\mathbf{v}(\mathbf{z}, \mathbf{\gamma
}_{0}) (w(\mathbf{z}, \tilde{\mathbf{\gamma}}, \vartheta^{\ast}) -w(\mathbf{z}, \mathbf{\gamma}_{0}, \vartheta^{\ast})))\\
&  = E_{\mathbf{\gamma}_{0}} (\mathbf{v}(\mathbf{z}, \mathbf{\gamma}_{0}) \nabla_{\mathbf{\gamma}} w(\mathbf{z}, \mathbf{\gamma}, \vartheta^{\ast})|_{\mathbf{\gamma} = \mathbf{\gamma}^{\ast}}^{\top}\sqrt{n}(\tilde{\mathbf{\gamma}} - \mathbf{\gamma})),
\end{align*}
where $\mathbf{\gamma}^{\ast}$ is between $\tilde{\mathbf{\gamma}}$ and $\mathbf{\gamma}_{0}$. The second term is such that
\begin{align*}
\lim_{n \rightarrow\infty} \nabla_{\mathbf{\gamma}} w(\mathbf{z},\mathbf{\gamma}, \vartheta^{\ast}) |_{\mathbf{\gamma} = \mathbf{\gamma}^{\ast}}  &  = \nabla_{\mathbf{\gamma}} \lim_{n \rightarrow\infty} w(\mathbf{z}, \mathbf{\gamma}, \vartheta^{\ast}) |_{\mathbf{\gamma} = \mathbf{\gamma}^{\ast}} \\
&  = \nabla_{\mathbf{\gamma}} w(\mathbf{z}, \mathbf{\gamma}, 0)|_{\mathbf{\gamma} = \mathbf{\gamma}_{0}},
\end{align*}
but $w(\mathbf{z}, \mathbf{\gamma}, 0) \equiv1$ for all $\mathbf{\gamma}$ and hence $\nabla_{\mathbf{\gamma}} w(\mathbf{z}, \mathbf{\gamma}, 0) \equiv0$ for all $\mathbf{\gamma}$.
\end{proof}

\begin{lemma} \label{LemSMAsymptotic1TML2} 
Let $\mathbf{z}_{i}$, $1\leq i\leq n$, be i.i.d. random vectors in $\mathbb{R}^{q}$ with distribution $G$ and $\mathbf{\zeta}$ a parameter in $\mathbb{R}^{s}$. Let $f(\mathbf{z},\mathbf{\zeta}): \mathbb{R}^{q}\mathbb{R}^{s}\rightarrow\mathbb{R}$ be a function continuously differentiable with respect to $\mathbf{\zeta}$, such that (i) $E_{G}(f(\mathbf{z},\mathbf{\zeta}))=0$ for all $\mathbf{\zeta}$, (ii) $E_{G}(f^{2}(\mathbf{z},\mathbf{\zeta}_{0}))<\infty$, and (iii) there exist $\delta>0$ and $c(\mathbf{z}):\mathbb{R}^{q}\rightarrow\mathbb{R}$ satisfying
\begin{equation*}
\sup_{\| \mathbf{\zeta}-\mathbf{\zeta}_{0} \| \leq\delta}\left\vert \frac{\partial f(\mathbf{z},\mathbf{\zeta})}{\partial\mathbf{\zeta}}\right\vert \leq c(\mathbf{z}),
\end{equation*}
where $E_{G}(c^{2}(\mathbf{z}))<\infty$. Let $\mathbf{\zeta}_{n}$ be a sequence of random variables converging to $\mathbf{\zeta}_{0}$ a.s.. Then,
\begin{equation}
\label{equSMlem2}n^{-1/2}\sum_{i=1}^{n}(f(\mathbf{z}_{i},\mathbf{\zeta}_{n})-f(\mathbf{z_{i}}, \mathbf{\zeta}_{0})) \overset{P}{\rightarrow} 0 .
\end{equation}
\end{lemma}

\begin{proof}[Proof]
Take $\varepsilon>0$ and consider the sequence of processes $S_{n}(\mathbf{\zeta})$, $\mathbf{\zeta} \in S = \left\{  \mathbf{\zeta}: \| \mathbf{\zeta} - \mathbf{\zeta}_{0} \| \leq\varepsilon\right\}$ in the space $C$ of the continuous functions in $S$:
\begin{equation*}
S_{n}(\mathbf{\zeta})=\frac{1}{n^{1/2}}\sum_{i=1}^{n} f(\mathbf{z}_{i}, \mathbf{\zeta}) .
\end{equation*}
To prove the Lemma is enough to show that $S_{n}(\mathbf{\zeta})$ is tight. According to Theorem 1 of \cite{jain1975}, $S_{n}(\mathbf{\zeta})$ converges in distribution to a Gaussian process $S(\lambda)$ and therefore it is tight.
\end{proof}
Finally, we report Lemma A3.1 of \citet{marazzi2009}
\begin{lemma}[Lemma A3.1 in \citet{marazzi2009}] \label{LemSMAsymptotic1TML3} 
Let $\mathbf{z}_{i}$, $1\leq i\leq n$, be i.i.d. random vectors in $\mathbb{R}^{q}$ with distribution $G$ and $\mathbf{\zeta}$ a parameter in $\mathbb{R}^{s}$. Let $f(\mathbf{z},\mathbf{\zeta}): \mathbb{R}^{q}\mathbb{R}^{s} \rightarrow\mathbb{R}^{k}$ be a continuous function, such that there exist $\delta>0$ and $c(\mathbf{z}):\mathbb{R}^{q}\rightarrow\mathbb{R}$ satisfying
\begin{equation*}
\sup_{\| \mathbf{\zeta}-\mathbf{\zeta}_{0} \| \leq\delta} \left\vert f(\mathbf{z},\mathbf{\zeta}) \right\vert \leq c(\mathbf{z}) \qquad a.s.,
\end{equation*}
and $E_{G}(c(\mathbf{z}))<\infty$. Let $\mathbf{\zeta}_{n}$ be a sequence of random variables converging to $\mathbf{\zeta}_{0}$ in probability. Then,
\begin{equation*}
\frac{1}{n} \sum_{i=1}^{n} f(\mathbf{z}_{i},\mathbf{\zeta}_{n}) \overset{P}{\rightarrow} E_{G} \left(  f(\mathbf{z_{i}}, \mathbf{\zeta}_{0})\right)
.
\end{equation*}
\end{lemma}

\begin{proof}[Proof of Theorem \ref{TeoSMAsymptotic1TMLreg}]
We recall that the 1TML estimator is defined as
\begin{equation*}
\hat{\mathbf{\gamma}}_{1} = \tilde{\mathbf{\gamma}} - J(\tilde{\mathbf{\gamma}},\vartheta^{\ast})^{-1} \frac{1}{n} \sum_{i=1}^{n} \left(  \mathbf{\eta}(\mathbf{z}_{i}, \tilde{\mathbf{\gamma}}, \tilde{\mathbf{\gamma}}, \vartheta^{\ast}) - \mathbf{h}(\tilde{\mathbf{\gamma}}, \tilde{\mathbf{\gamma}}, \tilde{\mathbf{\gamma}}, \vartheta^{\ast})\right)
\end{equation*}
We consider an expansion of $\mathbf{\eta}$ in its second argument around $\mathbf{\gamma}_{0}$ has follows
\begin{equation*}
\sum_{i=1}^{n} \mathbf{\eta}(\mathbf{z}_{i}, \tilde{\mathbf{\gamma}}, \tilde{\mathbf{\gamma}}, \vartheta^{\ast}) = \sum_{i=1}^{n} \mathbf{\eta}(\mathbf{z}_{i}, \mathbf{\gamma}_{0}, \tilde{\mathbf{\gamma}}, \vartheta^{\ast}) + J(\mathbf{\gamma}^{\ast},\vartheta^{\ast}) (\tilde{\mathbf{\gamma}}- \mathbf{\gamma}_{0})
\end{equation*}
and $\mathbf{\gamma}^{\ast}$ is between $\tilde{\mathbf{\gamma}}$ and $\mathbf{\gamma}_{0}$. Then,
\begin{align}
\sqrt{n} (\hat{\mathbf{\gamma}}_{1} - \mathbf{\gamma}_{0})  &  = \sqrt{n} (\tilde{\mathbf{\gamma}} - \mathbf{\gamma}_{0}) \nonumber \label{equSMAsymptoticTMLreg1}\\
& - n J(\tilde{\mathbf{\gamma}},\vartheta^{\ast})^{-1} \frac{1}{\sqrt{n}} \sum_{i=1}^{n} \left[  \mathbf{\eta}(\mathbf{z}_{i}, \mathbf{\gamma}_{0}, \tilde{\mathbf{\gamma}}, \vartheta^{\ast}) - \mathbf{h}(\tilde{\mathbf{\gamma}}, \tilde{\mathbf{\gamma}}, \tilde{\mathbf{\gamma}}, \vartheta^{\ast}) \right] \\
& - \sqrt{n} J(\tilde{\mathbf{\gamma}},\vartheta^{\ast})^{-1} J(\mathbf{\gamma}^{\ast},\vartheta^{\ast}) (\tilde{\mathbf{\gamma}} - \mathbf{\gamma}_{0})\nonumber\\
&  = - n J(\tilde{\mathbf{\gamma}},\vartheta^{\ast})^{-1} \frac{1}{\sqrt{n}} \sum_{i=1}^{n} \left[  \mathbf{\eta}(\mathbf{z}_{i}, \mathbf{\gamma}_{0}, \tilde{\mathbf{\gamma}}, \vartheta^{\ast}) - \mathbf{h}(\tilde{\mathbf{\gamma}}, \tilde{\mathbf{\gamma}}, \tilde{\mathbf{\gamma}}, \vartheta^{\ast}) \right] \nonumber\\
&  + (1 - J(\tilde{\mathbf{\gamma}},\vartheta^{\ast})^{-1} J(\mathbf{\gamma}^{\ast},\vartheta^{\ast})) \sqrt{n} (\tilde{\mathbf{\gamma}} - \mathbf{\gamma}_{0}).\nonumber
\end{align}
By Lemma \ref{LemSMAsymptotic1TML1}, we have
\begin{align}
& \frac{1}{\sqrt{n}} \sum_{i=1}^{n} \left[  \mathbf{\eta}(\mathbf{z}_{i}, \mathbf{\gamma}_{0}, \tilde{\mathbf{\gamma}}, \vartheta^{\ast}) - \mathbf{h}(\tilde{\mathbf{\gamma}}, \tilde{\mathbf{\gamma}}, \tilde{\mathbf{\gamma}}, \vartheta^{\ast}) \right] \nonumber\\
&  = \frac{1}{\sqrt{n}} \sum_{i=1}^{n} \left[  \mathbf{\eta}(\mathbf{z}_{i},\mathbf{\gamma}_{0}, \tilde{\mathbf{\gamma}}, \vartheta^{\ast}) -\mathbf{h}(\mathbf{\gamma}_{0}, \tilde{\mathbf{\gamma}}, \mathbf{\gamma}_{0}, \vartheta^{\ast}) \right]  + o_{p}(1).\nonumber
\end{align}
Let $\Sigma(\mathbf{\gamma}_{0})$ be the asymptotic covariance matrix of the maximum likelihood estimators and by Lemma \ref{LemSMAsymptotic1TML2}
\begin{equation} \label{equSMAsymptoticTMLreg2}
\frac{1}{\sqrt{n}} \sum_{i=1}^{n} \left[ \mathbf{\eta}(\mathbf{z}_{i}, \mathbf{\gamma}_{0}, \tilde{\mathbf{\gamma}}, \vartheta^{\ast}) - \mathbf{h}(\mathbf{\gamma}_{0}, \tilde{\mathbf{\gamma}}, \mathbf{\gamma}_{0}, \vartheta^{\ast}) \right]  \overset{\mathcal{L}}{\rightarrow} N(\mathbf{0},\Sigma(\mathbf{\gamma}_{0})).
\end{equation}
Finally, by Lemma \ref{LemSMAsymptotic1TML3} we have
\begin{equation} \label{equSMAsymptoticTMLreg3}
\mathbf{J}(\tilde{\mathbf{\gamma}},\vartheta^{\ast}) = \mathbf{G}(\mathbf{\gamma}_{0}) + o_{P}(1)
\end{equation}
and $\mathbf{J}(\mathbf{\gamma}^{\ast},\vartheta^{\ast}) = \mathbf{G}(\mathbf{\gamma}_{0}) + o_{P}(1)$. Hence, $(1 - \mathbf{J}(\tilde{\mathbf{\gamma}},\vartheta^{\ast})^{-1} \mathbf{J}(\mathbf{\gamma}^{\ast},\vartheta^{\ast})) \overset{P}{\rightarrow} 0$ by Slutsky Theorem. The asymptotic normality of the 1TML estimator follows from (\ref{equSMAsymptoticTMLreg1}), (\ref{equSMAsymptoticTMLreg2}), (\ref{equSMAsymptoticTMLreg3}) and the Slutsky theorem.
\end{proof}
We now consider the case without covariates.  Next Theorem follows immediately from Theorem \ref{TeoSMAsymptotic1TMLreg}.
\begin{theorem} \label{TeoSMAsymptotic1TML} 
Let us consider observations $(y_{i}^{\ast}, \delta_{i})$ ($1 \le i \le n$) following a GLG model with parameter $\mathbf{\theta}_{0} = (\mu_{0},\sigma_{0},\lambda_{0})$. Let $\tilde{\mathbf{\theta}}$ be an initial estimator. Assume: (i) $n^{1/2}(\tilde{\mathbf{\theta}} - \mathbf{\theta}_{0})$ is bounded in probability; (ii) the matrix $\mathbf{G}(\mathbf{\theta})$ is non singular; (iii) the weight function $w$ satisfies \ref{assA2}; (iv) the function $w$ is continuously differentiable. Then, the 1TML estimator $\hat{\mathbf{\theta}}_{1}$ is such that:
\begin{equation*}
\sqrt{n}(\hat{\mathbf{\theta}}_{1} - \mathbf{\theta}_{0}) \overset{\mathcal{L}}{\rightarrow} N(0, \mathbf{G}(\mathbf{\theta}_{0})^{-1} \mathbf{M}(\mathbf{\theta}_{0})\mathbf{G}(\mathbf{\theta}_{0})^{-\top} )
\end{equation*}
\end{theorem}
We recall that $\mathbf{M}(\mathbf{\theta}) = E\left(  \mathbf{v}(y,\delta,\mathbf{\theta}) \mathbf{v}(y,\delta,\mathbf{\theta})^{\top}\right)$ and $\mathbf{G}(\mathbf{\theta}) = E\left(  \nabla_{\mathbf{\theta}} \mathbf{v}(y,\delta,\mathbf{\theta}) \right)  $ and $\mathbf{v}(y,\delta,\mathbf{\theta}) $ is defined in equation (\ref{vscore}).
\begin{remark}
In practice, especially when the sample size is not very large, and $\vartheta^{\ast}$ is far from the asymptotic value $0$, a better approximation of the covariance matrix of the 1TML estimator can be based on the following sandwich formula. Let
\begin{align*}
&  \mathbf{\Lambda}(\mathbf{z}, \hat{\mathbf{\gamma}}_{1},\vartheta^{\ast})\\
&  = (\mathbf{\eta}(\mathbf{z}, \hat{\mathbf{\gamma}}_{1},\hat{\mathbf{\gamma}}_{1},\vartheta^{\ast}) - \mathbf{h}(\hat{\mathbf{\gamma}}_{1},\hat{\mathbf{\gamma}}_{1},\hat{\mathbf{\gamma}}_{1},\vartheta^{\ast})) (\mathbf{\eta}(\mathbf{z}, \hat{\mathbf{\gamma}}_{1},\hat{\mathbf{\gamma}}_{1},\vartheta^{\ast}) - \mathbf{h}(\hat{\mathbf{\gamma}}_{1},\hat{\mathbf{\gamma}}_{1},\hat{\mathbf{\gamma}}_{1},\vartheta^{\ast}))^{\top}.
\end{align*}
Then, the covariance matrix $\mathbf{\Sigma}(\hat{\mathbf{\gamma}}_{1})$ of the 1TML estimator can be estimated by
\begin{equation*}
\hat{\mathbf{\Sigma}}(\hat{\mathbf{\gamma}}_{1}) = \mathbf{J}(\hat{\mathbf{\gamma}}_{1}, \vartheta^{\ast})^{-1} \left(  \frac{1}{n} \sum_{i=1}^{n} \mathbf{\Lambda}(\mathbf{z}, \hat{\mathbf{\gamma}}_{1},\vartheta^{\ast}) \right)  \mathbf{J}(\hat{\mathbf{\gamma}}_{1}, \vartheta^{\ast})^{-\top} .
\end{equation*}
\end{remark}
\begin{remark}
For the case without covariates
\begin{equation*}
\mathbf{G}(\mathbf{\theta}) = E(\delta\mathbf{G}_{d} + (1-\delta) \mathbf{G}_{s}),
\end{equation*}
where the matrix $\mathbf{G}_{d}$ corresponds to the uncensored observations and $\mathbf{G}_{s}$ to the censored observations. Hence, the elements of $\mathbf{G}_{d}$ are
\begin{align*}
g_{d11}  &  = \frac{\partial d_{1}}{\partial\mu} = -\frac{1}{\sigma^{2}}%
\xi_{\lambda}^{\prime}(u) ,\\
g_{d12}  &  = \frac{\partial d_{1}}{\partial\sigma} = -\frac{1}{\sigma^{2}}\left(  \xi_{\lambda}(u)+\xi_{\lambda}^{\prime}(u)u\right)  ,\\
g_{d13}  &  = \frac{\partial d_{1}}{\partial\lambda} = \frac{1}{\sigma}\dot{\xi}_{\lambda}(u) ,
\end{align*}
\begin{align*}
g_{d21}  &  = \frac{\partial d_{2}}{\partial\mu} = -\frac{1}{\sigma^{2}}\left(  \xi_{\lambda}^{\prime}(u)u+\xi_{\lambda}(u)\right)  ,\\
g_{d22}  &  = \frac{\partial d_{2}}{\partial\sigma} = -\frac{1}{\sigma^{2}}\left(  2\xi_{\lambda}(u)u+\xi_{\lambda}^{\prime}(u)u^{2}+1\right)  ,\\
g_{d23}  &  = \frac{\partial z_{2}}{\partial\lambda} = \frac{1}{\sigma}\dot{\xi}_{\lambda}(u) u ,
\end{align*}
\begin{align*}
g_{d31}  &  = \frac{\partial z_{3}}{\partial\mu} = -\frac{1}{\sigma}\psi_{\lambda}^{\prime}(u),\\
g_{d32}  &  = \frac{\partial z_{3}}{\partial\sigma} = -\frac{1}{\sigma}\psi_{\lambda}^{\prime}(u)u,\\
g_{d32}  &  = \frac{\partial z_{3}}{\partial\lambda}=\dot{\psi}_{\lambda}(u) .
\end{align*}
\end{remark}

\subsection{The weighted likelihood estimator}
\label{SMSec1SWL} 
In the next two subsections we extend the results in \citet{agostinelli2014a} concerning the 1SWL estimator to the case with censored observations and to the regression case.

\subsubsection{The case without covariables}
\label{SecSM1SWL} 
\citet{markatou1998} introduced the weighted likelihood estimators for the case of non censored data in continuous models. Here, we extend this estimator to the case of censored data. Assume that $w(y,\mathbf{\theta})$ is a given weight function and that an initial highly robust and consistent but not necessarily efficient estimator $\tilde{\mathbf{\theta}}$ of $\mathbf{\theta}_{0}$ -- e.g., a TQ$\tau$ estimator $\tilde{\mathbf{\theta}}$ defined in Section \ref{SecTQtau} -- is available. Then, we define the \textit{weighted likelihood} (WL) \textit{estimator} for censored observations as a solution of the equation
\begin{equation}
\mathbf{u}(\mathbf{y}^{\ast},\mathbf{\delta},\mathbf{\theta}) = \sum_{i=1}^{n} w(y_{i}^{\ast},\mathbf{\theta})\mathbf{v}(y_{i}^{\ast},\delta_{i}, \mathbf{\theta}) = \mathbf{0}, \label{equSMWL1}
\end{equation}
where $\mathbf{v}(y_{i}^{\ast},\delta_{i},\mathbf{\theta})$ is the same as in (\ref{equML1}). When $w(y^{\ast},\mathbf{\theta}) \equiv1$, (\ref{equSMWL1}) coincides with (\ref{equML1}). Following \citet{markatou1998} we define the weight function as
\begin{equation}
w(y,\mathbf{\theta})=\min\left(  1,\frac{\left[  A(r_{p}(y,\mathbf{\theta}))+1\right]  ^{+}}{r_{p}(y,\mathbf{\theta})+1}\right)  , \label{Weights}
\end{equation}
where $[x]^{+}=\max(0,x)$, $r_{p}(y,\mathbf{\theta})$ is a convenient variant of \textit{Pearson residual}, measuring the agreement between the data and the assumed model, and the function $A(\cdot)$ is a \textit{residual adjustment function} \citep{lindsay1994}.

In order to define the Pearson residual, we proceed as follows. Let $H_{n,\mathbf{\theta}}$ be the semiparametric cdf (\ref{Hemp}) and $f_{n,\mathbf{\theta}}^{\ast}(y)=\int k(y,t,h) dH_{n,\mathbf{\theta}}$ be a semiparametric kernel density estimator of $f_{\mathbf{\theta}}(y)$ with bandwidth $h$. ($f_{n,\mathbf{\theta}}^{\ast}$ can be approximately computed by fitting a kernel density estimator to $k$ quantiles of $H_{n,\mathbf{\theta}}$, e.g., $k=1000$.) Let $f_{\mathbf{\theta}}^{\ast}(y)=\int k(y,t,h) f_{\mathbf{\theta}}(t) dt$ be the corresponding smoothed model density. Then, the Pearson residuals are defined by
\begin{equation*}
r_{p}(y,\mathbf{\theta})=\left[  f_{n,\mathbf{\theta}}^{\ast}(y)-f_{\mathbf{\theta}}^{\ast}(y)\right]  /f_{\mathbf{\theta}}^{\ast}(y).
\end{equation*}
Equation (\ref{equSMWL1}) can be solved using an iterative algorithm as in \citet{markatou1998}.

Following \citet{agostinelli1998} and \citet{agostinelli2014a}, we obtain the \textit{one-step weighted likelihood} (1SWL) \textit{estimator} for censored observations by applying one Newton-Raphson iteration to equation (\ref{equSMWL1}). This estimator turns out to be
\begin{equation} \label{equSM1SWL}
\mathbf{\hat{\theta}} = \mathbf{\tilde{\theta}} - \mathbf{J}^{-1}\mathbf{u}(\mathbf{y}^{\ast},\mathbf{\delta},\tilde{\mathbf{\theta}}),
\end{equation}
where $\mathbf{J} =\sum_{j=1}^{n}w(y_{j}^{\ast}, \mathbf{\tilde{\theta}}) \nabla_{\mathbf{\theta}} \mathbf{v}(y_{j}^{\ast}, \delta_{j}, \mathbf{\theta})|_{\mathbf{\theta}=\tilde{\mathbf{\theta}}}$ is the Jacobian matrix of the functions defining the estimating equations and $\nabla_{\mathbf{\theta}}$ denotes differentiation with respect to $\mathbf{\theta}$.

\subsubsection{The case with covariables}
\label{SecSM1SWLreg} 
For the model (\ref{model}) and parameters $\mathbf{\gamma} = (\mathbf{\theta}, \mathbf{\beta})$, we define the fully iterated WL estimator as the solution of the equations
\begin{equation*}
\mathbf{u}(X, \mathbf{y}^{\ast}, \mathbf{\delta}, \mathbf{\gamma}) = \sum_{j=1}^{n}\check{w}(\mathbf{x}_{j},y_{j}^{\ast},\mathbf{\gamma}) \mathbf{v}(\mathbf{x}_{j},y_{j}^{\ast},\delta_{j},\mathbf{\gamma}) = \mathbf{0},
\end{equation*}
where
\begin{equation*}
\check{w}(\mathbf{x}_{j},y_{j}^{\ast},\mathbf{\gamma})=w(y_{j}^{\ast}-\mathbf{x}_{j}^{\top}\mathbf{\beta},\mathbf{\theta}),
\end{equation*}
and $w$ is given by (\ref{Weights}). Besides, the \textit{1SWL regression estimator} is given by
\begin{equation*}
\mathbf{\hat{\gamma}} = \mathbf{\tilde{\gamma}} - \mathbf{J}^{-1} \mathbf{u}(X, \mathbf{y}^{\ast}, \mathbf{\delta}, \tilde{\mathbf{\gamma}})
\end{equation*}
and $\mathbf{J}=\sum_{i=1}^{n}\check{w}(\mathbf{x}_{i},y_{i},\mathbf{\tilde{\gamma}}) \nabla_{\mathbf{\gamma}} \mathbf{v}(\mathbf{x}_{i},y_{i}^{\ast},\tilde{\mathbf{\gamma}})|_{\mathbf{\gamma}=\tilde{\mathbf{\gamma}}}$ is a $(3+p) \times(3+p)$ Jacobian matrix, The weights are based on Pearson residuals comparing $F_{\tilde{\mathbf{\theta}}^{\ast}}$ to the estimator $H_{n,\tilde{\mathbf{\theta}}^{\ast}}$ given by (\ref{Hemp}), where the observations are $(\omega_{i}^{\ast},\delta_{i})=(y_{i}^{\ast}-\mathbf{x}_{i}^{\top}\mathbf{\bar{\beta}},\delta_{i})$, $1\leq i\leq n$.

%%%\bibliography{loggammacensarxiv}

\end{document}